\documentclass[a4paper]{article}
\usepackage[utf8]{inputenc}
\usepackage[T1]{fontenc}
\usepackage{lmodern}
\usepackage[english]{babel}
\usepackage[unicode]{hyperref}
\usepackage{graphics}
\usepackage{stmaryrd}
\usepackage{slashbox}
\usepackage{mathpartir}
\usepackage{amsmath}
\usepackage{amsthm}
\usepackage[all,2cell,ps]{xy}
\xyoption{curve}
\UseAllTwocells
\SilentMatrices
\CompileMatrices

\newcommand{\catquot}[1]{#1/\!\!\!\equiv}

\newcommand{\strid}[1]{\cpdfinput{#1.ps}}
\newcommand{\sstrid}[1]{\strid{#1}}
\newcommand{\svxym}[1]{\vcenter{\xymatrix@C=7ex@R=7ex{#1}}}

\newcommand{\axioms}{\text{\emph{Ax}}}
\newcommand{\sgramor}{\ \ \ |\ \ \ }

\allowdisplaybreaks[1]

\usepackage{shortcuts}
\renewcommand{\wrt}{wrt\xspace}

\title{The Structure of First-Order Causality}

\hypersetup{
  pdftitle={\csname @title\endcsname},
  pdfauthor={Samuel Mimram},
  unicode=true,
  colorlinks=true,
  linkcolor=black,
  citecolor=black,
  urlcolor=black
}

\author{Samuel Mimram\thanks{This work was has been supported by the CHOCO
    (``Curry Howard pour la Concurrence'', \hbox{ANR-07-BLAN-0324}) French ANR
    project.}}
\date{}

\newcommand{\FinSet}{\mathbf{FinSet}}
\newcommand{\FMR}{\mathbf{MRel}}

\newcommand{\Games}{\mathbf{Games}}
\newcommand{\CGames}{\mathbf{CGames}}
\newcommand{\Alg}[2]{\mathbf{Alg}_{#1}^{#2}}
\newcommand{\intset}[1]{{#1}}

\newcommand{\moves}[1]{M_{#1}}

\newcommand{\qforall}[1]{\forall{#1}.}
\newcommand{\qexists}[1]{\exists{#1}.}

\newtheorem{definition}{Definition}

\newtheorem{lemma}[definition]{Lemma}
\newtheorem{property}[definition]{Property}

\newtheorem{theorem}[definition]{Theorem}

\newtheorem{example}[definition]{Example}

\theoremstyle{remark}
\newtheorem{remark}[definition]{Remark}

\renewcommand{\paragraph}[1]{\bigskip\noindent\textbf{#1}\;}

\newcommand{\eqth}[1]{\mathfrak{#1}}
\newcommand{\intp}[1]{\llbracket{#1}\rrbracket}
\newcommand{\card}[1]{\left|#1\right|}
\newcommand{\size}[1]{\left|#1\right|}
\newcommand{\lrule}[1]{\text{(#1)}}
\newcommand{\lrulel}[1]{\lrule{$#1$-L}}
\newcommand{\lruler}[1]{\lrule{$#1$-R}}
\newcommand{\before}{\varolessthan}

\begin{document}
\maketitle

\begin{abstract}
  Game semantics describe the interactive behavior of proofs by interpreting
  formulas as games on which proofs induce strategies. Such a semantics is
  introduced here for capturing dependencies induced by quantifications in
  first-order propositional logic. One of the main difficulties that has to be
  faced during the elaboration of this kind of semantics is to characterize
  definable strategies, that is strategies which actually behave like a
  proof. This is usually done by restricting the model to strategies satisfying
  subtle combinatorial conditions, whose preservation under composition is often
  difficult to show. Here, we present an original methodology to achieve this
  task, which requires to combine advanced tools from game semantics, rewriting
  theory and categorical algebra. We introduce a diagrammatic presentation of
  the monoidal category of definable strategies of our model, by the means of
  generators and relations: those strategies can be generated from a finite set
  of atomic strategies and the equality between strategies admits a finite
  axiomatization, this equational structure corresponding to a polarized
  variation of the notion of bialgebra. This work thus bridges algebra and
  denotational semantics in order to reveal the structure of dependencies
  induced by first-order quantifiers, and lays the foundations for a mechanized
  analysis of causality in programming languages.
\end{abstract}

Denotational semantics were introduced to provide useful abstract invariants of
proofs and programs modulo cut-elimination or reduction. In particular, game
semantics, introduced in the nineties, have been very successful in capturing
precisely the interactive behavior of programs. In these semantics, every type
is interpreted as a \emph{game} (that is as a set of \emph{moves} that can be
played during the game) together with the rules of the game (formalized by a
partial order on the moves of the game indicating the dependencies between
them). Every move is to be played by one of the two players, called
\emph{Proponent} and \emph{Opponent}, who should be thought respectively as the
program and its environment. The interactions between these two players are
sequences of moves respecting the partial order of the game, called
\emph{plays}. In this setting, a program is characterized by the set of plays
that it can exchange with its environment during an execution and thus defines a
\emph{strategy} reflecting the interactive behavior of the program inside the
game specified by the type of the program.

The notion of \emph{pointer game}, introduced by Hyland and
Ong~\cite{hyland-ong:full-abstraction-pcf}, gave one of the first fully abstract
models of PCF (a simply-typed \hbox{$\lambda$-calculus} extended with recursion,
conditional branching and arithmetical constants). It has revealed that PCF
programs generate strategies with partial memory, called \emph{innocent} because
they react to Opponent moves according to their own \emph{view} of the
play. Innocence
is in this setting the main ingredient to characterize \emph{definable}
strategies, that is strategies which are the interpretation of a PCF term,
because it describes the behavior of the purely functional core of the language
(\ie $\lambda$-terms), which also corresponds to proofs in propositional
logic. This seminal work has lead to an extremely successful series of
semantics: by relaxing in various ways the innocence constraint on strategies,
it became suddenly possible to generalize this characterization to PCF programs
extended with imperative features such as references, control, non-determinism,
etc.

Unfortunately, these constraints are quite specific to game semantics and remain
difficult to link with other areas of computer science or algebra. They are
moreover very subtle and combinatorial and thus sometimes difficult to work
with. This work is an attempt to find new ways to describe the behavior of
proofs.


\paragraph{Generating instead of restricting.}
In this paper, we introduce a game semantics capturing dependencies induced by
quantifiers in first-order propositional logic, forming a strict monoidal
category called $\Games$. Instead of characterizing definable strategies of the
model by restricting to strategies satisfying particular conditions, we show
here that we can equivalently use a kind of converse approach. We show how to
\emph{generate} definable strategies by giving a \emph{presentation} of those
strategies: a finite set of definable strategies can be used to generate all
definable strategies by composition and tensoring, and the equality between
strategies obtained this way can be finitely axiomatized.

What we mean precisely by a presentation is a generalization of the usual notion
of presentation of a monoid to monoidal categories. For example, consider the
additive monoid $\mathbb{N}^2=\mathbb{N}\times\mathbb{N}$. It admits the
presentation~$\pangle{\;p,q\;|\;qp=pq\;}$, where $p$ and $q$ are two
\emph{generators} and $qp=pq$ is a relation between two elements of the free
monoid $M$ on $\{p,q\}$. This means that $\mathbb{N}^2$ is isomorphic to the
free monoid $M$ on the two generators, quotiented by the smallest congruence
$\equiv$ (\wrt{} multiplication) such that $qp\equiv pq$.
More generally, a (strict) monoidal category~$\mathcal{C}$ (such as~$\Games$)
can be presented by a \emph{polygraph}, consisting of typed generators in
dimension 1 and~2 and relations in dimension 3, such that the
category~$\mathcal{C}$ is monoidally equivalent to the free monoidal category on
the generators, quotiented by the congruence generated by the relations.


\paragraph{Reasoning locally.}
The usefulness of our construction is both theoretic and practical. It reveals
that the essential algebraic structure of dependencies induced by quantifiers is
a polarized variation of the well-known structure of bialgebra, thus bridging
game semantics and algebra. It also proves very useful from a technical point of
view: this presentation allows us to reason locally about strategies. In
particular, it enables us to deduce a posteriori that these strategies actually
\emph{compose}, which is not trivial, and it also enables us to deduce that the
strategies of the category~$\Games$ are \emph{definable} (one only needs to
check that generators are definable). Finally, the presentation gives a finite
description of the category, that we can hope to manipulate with a computer,
paving the way for a series of new tools to automate the study of semantics of
programming languages.

\paragraph{A game semantics capturing first-order causality.}
Game semantics has revealed that proofs in logic describe particular strategies
to explore formulas, or more generally sequents. Namely, a formula (or a
sequent) is a syntactic tree expressing in which order its connectives must be
introduced in cut-free proofs. In this sense, it can be seen as the rules of a
game whose moves correspond to connectives. For instance, consider a sequent of
the form
\begin{equation}
  \label{eq:ex-formula}
  \qforall x P\quad\vdash\quad\qforall y\qexists z Q
\end{equation}
%
where $P$ and $Q$ are propositional formulas which may contain free
variables. When searching for a proof of~\eqref{eq:ex-formula}, the $\forall y$
quantification must be introduced before the $\exists z$ quantification, and the
$\forall x$ quantification can be introduced independently. Here, introducing an
existential quantification on the right of a sequent should be thought as
playing a Proponent move (the strategy gives a witness for which the formula
holds) and introducing an universal quantification as playing an Opponent move
(the strategy receives a term from its environment, for which it has to show
that the formula holds); introducing a quantification on the left of a sequent
is similar but with polarities inverted since it is the same as introducing the
dual quantification on the right of the sequent. So, the game associated to the
formula~\eqref{eq:ex-formula} will be the partial order on the first-order
quantifications appearing in the formula, depicted below (to be read from the
top to the bottom):
\begin{equation}
  \label{eq:ex-formula-game}
  \xymatrix@R=4ex@C=4ex{
    \forall x&\ar@{-}[d]\forall y\\
    &\exists z\\
  }
\end{equation}
This partial order is sometimes called the \emph{syntactic partial order}
generated by the sequent. Possible proofs of sequent~\eqref{eq:ex-formula} in
first-order propositional logic are of one of the three following shapes:
\[
\inferrule{
\inferrule{
\inferrule{
\inferrule{\vdots}
{P[t/x]\vdash Q[t'/z]}
}
{P[t/x]\vdash\qexists z Q}
}
{P[t/x]\vdash \qforall y\qexists z Q}
}{\qforall x P\vdash \qforall y\qexists z Q}
\qquad\qquad
\inferrule{
\inferrule{
\inferrule{
\inferrule{\vdots}
{P[t/x]\vdash Q[t'/z]}
}
{P[t/x]\vdash\qexists z Q}
}
{\qforall x P\vdash \qexists z Q}
}{\qforall x P\vdash \qforall y \qexists z Q}
\qquad\qquad
\inferrule{
\inferrule{
\inferrule{
\inferrule{\vdots}
{P[t/x]\vdash Q[t'/z]}
}
{\qforall x P\vdash Q[t'/z]}
}
{\qforall x P\vdash \qexists z Q}
}{\qforall x P\vdash \qforall y \qexists z Q}
\]
where $P[t/x]$ denotes the formula $P$ where every occurrence of the free
variable $x$ has been replaced by the term $t$. These proofs introduce the
connectives in the orders depicted respectively below
\[
\xymatrix@R=4ex@C=4ex{
\ar@{-}[d]\forall x\\
\ar@{-}[d]\forall y\\
\exists z\\
}
\qquad
\xymatrix@R=4ex@C=4ex{
\ar@{-}[d]\forall y\\
\ar@{-}[d]\forall x\\
\exists z\\
}
\qquad
\xymatrix@R=4ex@C=4ex{
\ar@{-}[d]\forall y\\
\ar@{-}[d]\exists z\\
\forall x\\
}
\]
which are all total orders extending the partial order of the
game~\eqref{eq:ex-formula-game}: these correspond to the plays in the strategies
interpreting the proofs in the game semantics. In this sense, they have more
dependencies between moves: proofs add causal dependencies between connectives.

Some sequentializations induced by proofs are not really relevant. For example
consider a proof of the form
\[
\inferrule{
  \inferrule{
    \inferrule{\pi}
    {P \vdash Q}
  }
  {P\vdash\qforall y Q}
}
{\qexists x P\vdash\qforall y Q}
\]
The order in which the introduction rules of the universal and existential
quantifications are introduced is not really significant here since this proof
might always be reorganized into the proof
\[
\inferrule{
  \inferrule{
    \inferrule{\pi}
    {P\vdash Q}
  }
  {\qexists x P\vdash Q}
}
{\qexists x P\vdash\qforall y Q}
\]
by ``permuting'' the introduction rules. Similarly, the following permutations
of rules are always possible:
\[
\inferrule{
  \inferrule{
    \inferrule{\pi}
    {P[t/x]\vdash Q[u/y]}
  }
  {P[t/x]\vdash \qexists y Q}
}
{\qforall x P\vdash\qexists y Q}
\rightsquigarrow
\inferrule{
  \inferrule{
    \inferrule{\pi}
    {P[t/x]\vdash Q[u/y]}
  }
  {\qforall x P\vdash Q[u/y]}
}
{\qforall x P\vdash\qexists y Q}
\qtand
\inferrule{
  \inferrule{
    \inferrule{\pi}
    {P[t/x]\vdash Q}
  }
  {P[t/x]\vdash\qforall y Q}
}
{\qforall x P\vdash\qforall y Q}
\rightsquigarrow
\inferrule{
  \inferrule{
    \inferrule{\pi}
    {P[t/x]\vdash Q}
  }
  {\qforall x P\vdash Q}
}
{\qforall x P\vdash\qforall y Q}
\]
Interestingly, the permutation
\[
\inferrule{
  \inferrule{
    \inferrule{\pi}
    {P\vdash Q[t/y]}
  }
  {P\vdash\qexists y Q}
}
{\qexists x P\vdash\qexists y Q}
\qquad
\rightsquigarrow
\qquad
\inferrule{
  \inferrule{
    \inferrule{\pi}
    {P\vdash Q[t/y]}
  }
  {\qexists x P\vdash Q[t/y]}
}
{\qexists x P\vdash\qexists y Q}
\]
is only possible if the term $t$ used in the introduction rule of the $\exists
y$ connective does not have $x$ as free variable. If the variable~$x$ is free in
$t$ then the rule introducing $\exists y$ can only be used after the rule
introducing the $\exists x$ connective. Now, the sequent~$\qexists x
P\vdash\qexists y Q$ will be interpreted by the following game
\[
\strid{dep_ex0}
\]
Whenever the~$\exists y$ connective depends on the~$\exists x$ connective (\ie
whenever~$x$ is free in the witness term~$t$ provided for~$y$), the strategy
corresponding to the proof will contain a causal dependency, which will be
depicted by an oriented wire
\[
\strid{dep_ex}
\]
and we sometimes say that the move~$\exists x$ \emph{justifies} the
move~$\exists y$. A simple further study of permutability of introduction rules
of first-order quantifiers shows that this is the only kind of relevant
dependencies.
These permutations of rules where the motivation for the introduction of
non-alternating asynchronous game semantics~\cite{mellies-mimram:ag5}, where
plays are considered modulo certain permutations of consecutive moves.
However, we focus here on causality and define strategies by the dependencies
they induce on moves (a precise description of the relation between these two
points of view was investigated in~\cite{mimram:phd}).
They are also very closely related to the motivations for the introduction of
Hintikka's games and independence friendly logic~\cite{hintikka-sandu:gts}.

We thus build a strict monoidal category whose objects are games and whose
morphisms are strategies, in which we can interpret formulas and proofs in the
connective-free fragment of first-order propositional logic, and write~$\Games$
for the subcategory of definable strategies. One should thus keep in mind the
following correspondences while reading this paper:
\begin{center}
  \begin{tabular}{|c|c|c|c|}
    \hline
    category&logic&game semantics&combinatorial objects\\
    \hline
    object&formula&game&syntactic order\\
    morphism&proof&strategy&justification order\\
    \hline
  \end{tabular}
\end{center}
This paper is devoted to the construction of a presentation for this
category. We introduce formally the notion of presentation of a monoidal
category in Section~\ref{sec:pres} and recall some useful classical algebraic
structures in Section~\ref{sec:alg-struct}. Then, we give a presentation of the
category of relations in Section~\ref{section:presentation-rel} and extend this
presentation to the category~$\Games$, that we define formally in
Section~\ref{section:games-strategies}.


\section{Presentations of monoidal categories}
\label{sec:pres}
We recall here briefly some basic definitions in category theory. The interested
reader can find a more detailed presentation of these concepts in MacLane's
reference book~\cite{maclane:cwm}.

\paragraph{Monoidal categories.}
A \emph{monoidal category} $(\mathcal{C},\otimes,I)$ is a category $\mathcal{C}$
together with a functor
\[
\otimes:\mathcal{C}\times\mathcal{C}\to\mathcal{C}
\]
and natural isomorphisms
\[
\alpha_{A,B,C}:(A\otimes B)\otimes C\to A\otimes(B\otimes C)
\tcomma\quad
\lambda_A:I\otimes A\to A
\qtand
\rho_A:A\otimes I\to A
\]
satisfying coherence axioms~\cite{maclane:cwm}. A symmetric monoidal category
$\mathcal{C}$ is a monoidal category $\mathcal{C}$ together with a natural
isomorphism
\[
\gamma_{A,B}:A\otimes B\to B\otimes A
\]
satisfying coherence axioms and such that
$\gamma_{B,A}\circ\gamma_{A,B}=\id_{A\otimes B}$. A monoidal
category~$\mathcal{C}$ is \emph{strictly} monoidal when the natural isomorphisms
$\alpha$, $\lambda$ and $\rho$ are identities. For the sake of simplicity, in
the rest of this paper we only consider strict monoidal categories. Formally, it
can be shown that it is not restrictive, using MacLane's coherence
theorem~\cite{maclane:cwm}: every monoidal category is monoidally equivalent to
a strict one.

A (strict) \emph{monoidal functor} $F:\mathcal{C}\to\mathcal{D}$ between two
strict monoidal categories $\mathcal{C}$ and $\mathcal{D}$ is a functor $F$
between the underlying categories
such that \hbox{$F(A\otimes B)=F(A)\otimes F(B)$} for every objects $A$ and $B$
of $\mathcal{C}$, and $F(I)=I$.
A \emph{monoidal natural transformation} $\theta:F\to G$ between two monoidal
functors $F,G:\mathcal{C}\to\mathcal{D}$ is a natural transformation between the
underlying functors $F$ and $G$ such that $\theta_{A\otimes
  B}=\theta_A\otimes\theta_B$ for every objects $A$ and $B$ of $\mathcal{C}$,
and $\theta_I=\id_I$.
Two monoidal categories $\mathcal{C}$ and $\mathcal{D}$ are \emph{monoidally
  equivalent} when there exists a pair of monoidal functors
$F:\mathcal{C}\to\mathcal{D}$ and $G:\mathcal{D}\to\mathcal{C}$ and two
invertible monoidal natural transformations $\eta:\mathrm{Id}_\mathcal{C}\to GF$
and $\varepsilon:FG\to\mathrm{Id}_\mathcal{D}$.

\paragraph{Monoidal theories.}
A \emph{monoidal theory} $\mathbb{T}$ is a strict monoidal category whose
objects are the natural integers, such that the tensor product on objects is the
addition of integers. By an integer $\underline{n}$, we mean here the finite
ordinal $\underline{n}=\{0,1,\ldots,n-1\}$ and the addition is given by
$\underline{m}+\underline{n}=\underline{m+n}$ (we will simply write~$n$ instead
of~$\underline{n}$ in the following). An \emph{algebra} $F$ of a monoidal theory
$\mathbb{T}$ in a strict monoidal category $\mathcal{C}$ is a strict monoidal
functor from~$\mathbb{T}$ to $\mathcal{C}$; we
write~$\Alg{\mathbb{T}}{\mathcal{C}}$ for the category of algebras from
$\mathbb{T}$ to $\mathcal{C}$ and monoidal natural transformations between
them. Monoidal theories are sometimes called PRO, this terminology was
introduced by MacLane in~\cite{maclane:ca} as an abbreviation for ``category
with products''. They generalize equational theories -- or Lawere
theories~\cite{lawvere:phd} -- in the sense that operations are typed and can
moreover have multiple outputs as well as multiple inputs, and are not
necessarily cartesian but only monoidal.

\paragraph{Presentations of monoidal categories.}
\label{subsection:moncat-presentation}
We now recall the notion of \emph{presentation} of a monoidal category by the
means of typed 1- and 2-dimensional generators and relations.

Suppose that we are given a set $E_1$ whose elements are called \emph{atomic
  types} or \emph{generators for objects}. We write $E_1^*$ for the free monoid
on the set $E_1$ and $i_1:E_1\to E_1^*$ for the corresponding injection; the
product of this monoid is written $\otimes$.
The elements of~$E_1^*$ are called \emph{types}. Suppose moreover that we are
given a set $E_2$, whose elements are called \emph{generators} (\emph{for
  morphisms}), together with two functions $s_1,t_1:E_2\to E_1^*$, which to
every generator associate a type called respectively its \emph{source} and
\emph{target}. We call a \emph{signature} such a 4-uple $(E_1,s_1,t_1,E_2)$:
\[
\svxym{
  E_1\ar[d]_{i_1}&\ar@<-0.7ex>[dl]_{s_1}\ar@<0.7ex>[dl]^{t_1}E_2\\
  E_1^*&\\
}
\]
\noindent
Every such signature $(E_1,s_1,t_1,E_2)$ generates a free strict monoidal
category $\mathcal{E}$, whose objects are the elements of~$E_1^*$ and whose
morphisms are formal composite and formal tensor products of elements of $E_2$,
quotiented by suitable laws imposing associativity of composition and tensor and
compatibility of composition with tensor, see~\cite{burroni:higher-word}. If we
write $E_2^*$ for the morphisms of this category and \hbox{$i_2:E_2\to E_2^*$}
for the injection of the generators into this category, we get a diagram
\[
\svxym{
  E_1\ar[d]_{i_1}&\ar@<-0.7ex>[dl]_{s_1}\ar@<0.7ex>[dl]^{t_1}E_2\ar[d]_{i_2}\\
  E_1^*&\ar@<-0.7ex>[l]_{\overline{s_1}}\ar@<0.7ex>[l]^{\overline{t_1}}E_2^*\\
}
\]
in $\Set$ together with a structure of monoidal category~$\mathcal{E}$ on the
graph
\[
\xymatrix@C=10ex@R=10ex{
  E_1^*&\ar@<-0.7ex>[l]_{\overline{s_1}}\ar@<0.7ex>[l]^{\overline{t_1}}E_2^*\\
}
\]
where the morphisms $\overline{s_1},\overline{t_1}:E_2^*\to E_1^*$ are the
morphisms (unique by universality of $E_2^*$) such that
\hbox{$s_1=\overline{s_1}\circ i_2$} and \hbox{$t_1=\overline{t_1}\circ i_2$}.
The \emph{size} $\size{f}$ of a morphism $f:A\to B$ in $E_2^*$ is defined
inductively by
\[
\begin{array}{r@{ = }l@{\qquad}r@{ = }l}
  \size{\id}&0
  &
  \size{f}&1
  \text{\quad if $f$ is a generator}
  \\
  \size{f_1\otimes f_2}
  &
  \size{f_1}+\size{f_2}
  &
  \size{f_2\circ f_1}
  &
  \size{f_1}+\size{f_2}
\end{array}
\]
In particular, a morphism is of size $0$ if and only if it is an identity.

Our constructions are an instance in dimension 2 of Burroni's polygraphs
\cite{burroni:higher-word}, and Street's 2\nbd
computads~\cite{street:limit-indexed-by-functors}, who made precise the sense in
which the generated monoidal category is free on the signature. Namely, the
following notion of equational theory is a specialization of the definition of a
3-polygraph to the case where there is only one generator for 0-cells.


\begin{definition}
  A \textbf{monoidal equational theory} is a 7-uple
  \[
  \eqth{E}=(E_1,s_1,t_1,E_2,s_2,t_2,E_3)
  \]
  where $(E_1,s_1,t_1,E_2)$ is a signature together with a set $E_3$ of
  \emph{relations} and two morphisms $s_2,t_2:E_3\to E_2^*$, as pictured in the
  diagram
  \[
  \svxym{
    E_1\ar[d]_{i_1}&\ar@<-0.7ex>[dl]_{s_1}\ar@<0.7ex>[dl]^{t_1}E_2\ar[d]_{i_2}&\ar@<-0.7ex>[dl]_{s_2}\ar@<0.7ex>[dl]^{t_2}E_3\\
    E_1^*&\ar@<-0.7ex>[l]_{\overline{s_1}}\ar@<0.7ex>[l]^{\overline{t_1}}E_2^*\\
  }
  \]
  such that $\overline{s_1}\circ s_2=\overline{s_1}\circ t_2$ and
  $\overline{t_1}\circ s_2=\overline{t_1}\circ t_2$.
\end{definition}
Every equational theory defines a monoidal category
$\mathbb{E}=\catquot{\mathcal{E}}$ obtained from the monoidal category
$\mathcal{E}$ generated by the signature $(E_1,s_1,t_1,E_2)$ by quotienting the
morphisms by the congruence $\equiv$ generated by the relations of the
equational theory $\eqth{E}$: it is the smallest congruence (\wrt{} both
composition and tensoring) such that $s_2(e)\equiv t_2(e)$ for every element $e$
of $E_3$.

We say that a monoidal equational theory $\eqth{E}$ is a \emph{presentation} of
a strict monoidal category $\mathcal{M}$ when $\mathcal{M}$ is monoidally
equivalent to the category $\mathbb{E}$ generated by~$\eqth{E}$. Any monoidal
category~$\mathcal{M}$ admits a presentation (for example, the trivial
presentation with~$E_1$ the set of objects of~$\mathcal{M}$, $E_2$ the set of
morphisms of~$\mathcal{M}$, and~$E_3$ the set of all equalities between
morphisms holding in~$\mathcal{M}$), which is not unique in general.
In such a presentation, the category $\mathcal{E}$ generated by the signature
underlying~$\eqth{E}$ should be thought as a category of ``terms'' (which will
be considered modulo the relations described by~$E_2$) and is thus sometimes
called the \emph{syntactic category} of~$\eqth{E}$.

We sometimes informally say that an equational theory has a \emph{generator}
$f:A\to B$ to mean that $f$ is an element of $E_2$ such that $s_1(f)=A$ and
$t_1(f)=B$. We also say that the equational theory has a \emph{relation} $f=g$
to mean that there exists an element $e$ of $E_3$ such that $s_2(e)=f$
and~$t_2(e)=g$.


We say that two equational theories are \emph{equivalent} when they generate
monoidally equivalent categories.
A generator $f$ in an equational theory $\eqth{E}$ is \emph{superfluous} when
the equational theory $\eqth{E'}$ obtained from $\eqth{E}$ by removing the
generator $f$ and all equations involving $f$, is equivalent to
$\eqth{E}$. Similarly, an equation $e$ is \emph{superfluous} when the
equational theory $\eqth{E'}$ obtained from $\eqth{E}$ by removing the
equation $e$ is equivalent to $\eqth{E}$. An equational theory is
\emph{minimal} when it does not contain any superfluous generator or equation.

Notice that every monoidal equational theory $(E_1,s_1,t_1,E_2,s_2,t_2,E_3)$
where the set~$E_1$ is reduced to only one object $\{1\}$ generates a monoidal
category which is a monoidal theory ($\N$ is the free monoid on one object),
thus giving a notion of presentation of those categories.

\paragraph{Presented categories as models.}
Suppose that a strict monoidal category $\mathcal{M}$ is presented by an
equational theory $\eqth{E}$, generating a category
$\mathbb{E}=\catquot{\mathcal{E}}$. The proof that $\eqth{E}$
presents~$\mathcal{M}$ can generally be decomposed in two parts:
\begin{enumerate}
\item \emph{$\mathcal{M}$ is a model of the equational theory $\eqth{E}$}: there
  exists a functor \hbox{$M:\mathbb{E}\to\mathcal{M}$}. This amounts to checking
  that there exists a functor \hbox{$M':\mathcal{E}\to\mathcal{M}$} such that
  for all morphisms \hbox{$f,g:A\to B$} in $\mathcal{E}$, $f\equiv g$ implies
  $M'f=M'g$.
\item \emph{$\mathcal{M}$ is a fully-complete model of the equational theory
    $\eqth{E}$}: the functor $M$ is full and faithful.
\end{enumerate}
We sometimes say that a morphism $f:A\to B$ of $\mathbb{E}$ \emph{represents}
the morphism \hbox{$M{f}:M{A}\to M{B}$} of~$\mathcal{M}$.

Usually, the first point is a straightforward verification. Proving that the
functor $M$ is full and faithful often requires more work. In this paper, we use
the methodology introduced by Burroni~\cite{burroni:higher-word} and refined by
Lafont~\cite{lafont:boolean-circuits}. We first define \emph{canonical forms}
which are canonical representatives of the equivalence classes of morphisms of
$\mathcal{E}$ under the congruence $\equiv$ generated by the relations of
$\eqth{E}$. Proving that every morphism is equal to a canonical form can be done
by induction on the size of the morphisms. Then, we show that the functor $M$ is
full and faithful by showing that the canonical forms are in bijection with the
morphisms of~$\mathcal{M}$.

It should be noted that this is not the only technique to prove that an
equational theory presents a monoidal category. In particular, Joyal and Street
have used topological methods~\cite{joyal-street:geometry-tensor-calculus} by
giving a geometrical construction of the category generated by a signature, in
which morphisms are equivalence classes under continuous deformation of
progressive plane diagrams (we give some more details about those diagrams, also
called string diagrams, later on). Their work is for example extended by Baez
and Langford in~\cite{baez-langford:two-tangles} to give a presentation of the
2-category of 2-tangles in 4~dimensions. The other general methodology the
author is aware of, is given by Lack in~\cite{lack:composing-props}, by
constructing elaborate monoidal theories from simpler monoidal theories. Namely,
a monoidal theory can be seen as a monad in a particular span bicategory, and
monoidal theories can therefore be ``composed'' given a distributive law between
their corresponding monads. We chose not to use those methods because, even
though they can be very helpful to build intuitions, they are difficult to
formalize and even more to mechanize: we believe indeed that some of the tedious
proofs given in this paper could be somewhat automated, a first step in this
direction was given in~\cite{mimram:rta10} where we describe an algorithm to
compute critical pairs in polygraphic rewriting systems of dimension 2.

\paragraph{String diagrams.}
\label{subsection:string-diagrams}
\emph{String diagrams} provide a convenient way to represent and manipulate the
morphisms in the category generated by a presentation. Given an object $M$ in a
strict monoidal category $\mathcal{C}$, a morphism $\mu:M\otimes M\to M$ can be
drawn graphically as a device with two inputs and one output of type $M$ as
follows:
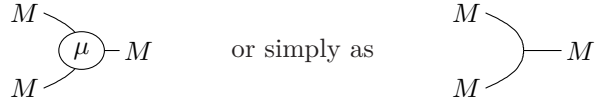
\[
\strid{mult_m_label}
\qquad\text{or simply as}\qquad
\strid{mult_m}
\]
when it is clear from the context which morphism of type $M\otimes M\to M$ we
are picturing (we sometimes even omit the source and target of the
morphisms). Similarly, the identity $\id_M:M\to M$ (which we sometimes simply
write $M$) can be pictured as a wire
\[
\strid{id_m}
\]
The tensor $f\otimes g$ of two morphisms $f:A\to B$ and \hbox{$g:C\to D$} is
obtained by putting the diagram corresponding to $f$ above the diagram
corresponding to~$g$.
So, for instance, the morphism $\mu\otimes M$
can be drawn diagrammatically as
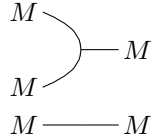
\[
\strid{mult_x_id_m}
\]
Finally, the composite $g\circ f:A\to C$ of two morphisms $f:A\to B$ and $g:B\to
C$ can be drawn diagrammatically by putting the diagram corresponding to $g$ at
the right of the diagram corresponding to $f$ and ``linking the wires''.
The diagram corresponding to the morphism $\mu\circ(\mu\otimes M)$
is thus
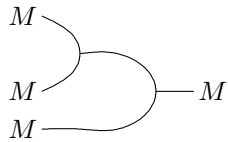
\[
\strid{mult_assoc_l_m}
\]

Suppose that $(E_1,s_1,t_1,E_2)$ is a signature. Every element $f$ of $E_2$ such
that
\[
s_1(f)=A_1\otimes\cdots\otimes A_m
\qtand
t_1(f)=B_1\otimes\cdots\otimes B_n
\]
where the $A_i$ and $B_i$ are elements of $E_1$, can be similarly represented by
a diagram
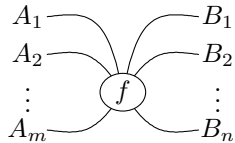
\[
\strid{signature_f}
\]
where wires correspond to generators for objects and circled points to
generators for morphisms. Bigger diagrams can be constructed from these diagrams
by composing and tensoring them, as explained above. Joyal and Street have shown
in details in~\cite{joyal-street:geometry-tensor-calculus} that the category of
those diagrams, modulo continuous deformations, is precisely the free category
generated by a signature (which they call a ``tensor scheme''). For example, the
equality
\[
(M\otimes\mu)\circ(\mu\otimes M\otimes M)
\qeq
(\mu\otimes M)\circ(M\otimes M\otimes\mu)
\]
in the category $\mathcal{C}$ of the above example, which holds because of the
axioms satisfied in any monoidal category, can be shown by continuously
deforming the diagram on the left-hand side below into the diagram on the
right-hand side:
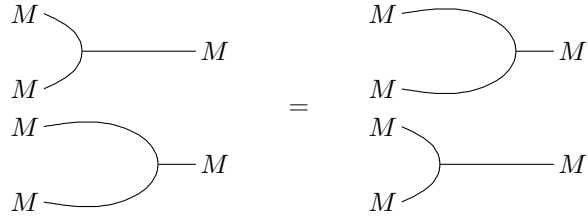
\[
\sstrid{mu_x_mu_r}
\qeq
\sstrid{mu_x_mu_l}
\]
All the equalities satisfied in any monoidal category generated by a signature
have a similar geometrical interpretation. And conversely, any deformation of
diagrams corresponds to an equality of morphisms in monoidal categories.

\section{Algebraic structures}
\label{sec:alg-struct}
In this section, we recall the categorical formulation of some well-known
algebraic structures, the most fundamental in this work being maybe the notion
of \emph{bialgebra}. We give those definitions in the setting of a strict
monoidal category which is \emph{not} required to be symmetric. We suppose that
$(\mathcal{C},\otimes,I)$ is a strict monoidal category, fixed throughout the
section.

\paragraph{Symmetric objects.}
A \emph{symmetric object} of $\mathcal{C}$ is an object~$S$ together with a
morphism
\[
\gamma:S\otimes S\to S\otimes S
\]
called \emph{symmetry} and pictured as
\begin{equation}
  \label{eq:sym-string}
  \strid{sym_s}
\end{equation}
such that the diagrams
\[
\xymatrix{
  S\otimes S\otimes S\ar[d]_{S\otimes\gamma}\ar[r]^{\gamma\otimes S}&S\otimes S\otimes S\ar[r]^{S\otimes\gamma}&S\otimes S\otimes S\ar[d]^{\gamma\otimes S}\\
  S\otimes S\otimes S\ar[r]_{\gamma\otimes S}&S\otimes S\otimes S\ar[r]_{S\otimes\gamma}&S\otimes S\otimes S\\
}
\]
and
\[
\xymatrix{
  &S\otimes S\ar[dr]^{\gamma}&\\
  S\otimes S\ar[ur]^{\gamma}\ar[rr]_{S\otimes S}&&S\otimes S\\
}
\]
commute. Graphically,
\begin{equation}
  \label{eq:sym}
  \begin{array}{rcl}
    \sstrid{yang_baxter_r}
    &\qeq&
    \sstrid{yang_baxter_l}
    \\
    \sstrid{sym_sym}
    &\qeq&
    \sstrid{id_x_id}
  \end{array}
\end{equation}
(the first equation is sometimes called the Yang-Baxter equation for braids). In
particular, in a symmetric monoidal category, every object is canonically
equipped with a structure of symmetric object.

\paragraph{Monoids.}
\label{subsection:monoids}
A \emph{monoid} $(M,\mu,\eta)$ in $\mathcal{C}$ is an object $M$ together with
two morphisms
\[
\mu : M\otimes M\to M
\qtand
\eta : I\to M
\]
called respectively \emph{multiplication} and \emph{unit} and pictured respectively as
\begin{equation}
  \label{eq:monoid-string}
  \strid{mult_m}
  \qtand
  \strid{unit_m}
\end{equation}
such that the diagrams
\[
\vxym{
  M\otimes M\otimes M\ar[d]_{M\otimes\mu}\ar[r]^-{\mu\otimes M}&M\otimes M\ar[d]^{\mu}\\
  M\otimes M\ar[r]_{\mu}&M
}
\tand
\vxym{
  \ar[dr]_{M}I\otimes M\ar[r]^{\eta\otimes M}&M\otimes M\ar[d]_{\mu}&\ar[l]_{M\otimes\eta}M\otimes I\ar[dl]^{M}\\
  &M&
}
\]
commute. Graphically,
\begin{equation}
  \label{eq:monoid}
  \begin{array}{c}
    \strid{mult_assoc_l}
    \qeq
    \strid{mult_assoc_r}
    \\
    \sstrid{mult_unit_l}
    \qeq
    \sstrid{mult_unit_c}
    \qeq
    \sstrid{mult_unit_r}
  \end{array}
\end{equation}

A \emph{symmetric monoid} is a monoid equipped with a symmetry morphism
\hbox{$\gamma:M\otimes M\to M\otimes M$} which is compatible with the operations
of the monoid in the sense that it makes the diagrams
\begin{equation*}
  \begin{array}{c}
    \xymatrix{
      \ar[d]_{\mu\otimes M}M\otimes M\otimes M\ar[r]^{M\otimes \gamma}&M\otimes M\otimes M\ar[r]^{\gamma\otimes M}&M\otimes M\otimes M\ar[d]^{M\otimes\mu}\\
      M\otimes M\ar[rr]_{\gamma}&&M\otimes M\\
    }
    \\[4ex]
    \xymatrix{
      \ar[d]_{M\otimes \mu}M\otimes M\otimes M\ar[r]^{\gamma\otimes M}&M\otimes M\otimes M\ar[r]^{M\otimes\gamma}&M\otimes M\otimes M\ar[d]^{\mu\otimes M}\\
      M\otimes M\ar[rr]_{\gamma}&&M\otimes M\\
    }
    \\[4ex]
    \xymatrix{
      &M\otimes M\ar[dr]^{\gamma}&\\
      M\ar[ur]^{\eta\otimes M}\ar[rr]_{\eta\otimes M}&&M\otimes M\\
    }
    \qquad
    \xymatrix{
      &M\otimes M\ar[dr]^{\gamma}&\\
      M\ar[ur]^{M\otimes\eta}\ar[rr]_{M\otimes\eta}&&M\otimes M\\
    }
    \\
  \end{array}
\end{equation*}
commute. Graphically,
\begin{equation}
  \label{eq:monoid-nat}
  \begin{array}{cc}
    \strid{mult_sym_rnat_r}
    \qeq
    \strid{mult_sym_rnat_l}
    \\
    \strid{eta_sym_rnat_l}
    \qeq
    \strid{eta_sym_rnat_r}
  \end{array}
\end{equation}
are satisfied, as well as the equations obtained by turning the diagrams
upside-down. A \emph{commutative monoid} is a symmetric monoid such that
the diagram
\[
\vxym{
  &M\otimes M\ar[dr]^{\mu}&\\
  M\otimes M\ar[ur]^{\gamma}\ar[rr]_{\mu}&&M
}
\]
commutes. Graphically,
\begin{equation}
  \label{eq:monoid_mult_comm}
  \strid{mult_comm}
  \qeq
  \strid{mult}
\end{equation}
In particular, a commutative monoid in a symmetric monoidal category is a
commutative monoid whose symmetry corresponds to the symmetry of the category:
$\gamma=\gamma_{M,M}$. In this case, the equations \eqref{eq:monoid-nat} can
always be deduced from the naturality of the symmetry of the monoidal category.

A \emph{comonoid} $(M,\delta,\varepsilon)$ in $\mathcal{C}$ is an object $M$
together with two morphisms
\[
\delta:M\to M\otimes M
\qtand
\varepsilon:M\to I
\]
respectively drawn as
\begin{equation}
  \label{eq:comonoid-string}
  \strid{comult_m}
  \qqtand
  \strid{counit_m}
\end{equation}
satisfying dual coherence diagrams. Similarly, the notions symmetric comonoid
and cocommutative comonoid can be defined by duality.

The definition of a monoid can be reformulated internally, in the language of
equational theories:
\begin{definition}
  \label{definition:e-t-monoid}
  The \emph{equational theory of monoids} $\eqth{M}$ has one generator for
  objects $1$ and two generators for morphisms $\mu:2\to 1$ and $\eta:0\to 1$
  subject to the three relations
  \begin{equation}
    \label{eq:monoid-theory}
    \begin{array}{c}
      \mu\circ(\mu\otimes\id_1)
      \qeq
      \mu\circ(\id_1\otimes\mu)
      \\
      \mu\circ(\eta\otimes\id_1)
      \qeq
      \id_1
      \qeq
      \mu\circ(\id_1\otimes\eta)
    \end{array}
  \end{equation}
\end{definition}
\noindent
The equations~\eqref{eq:monoid-theory} correspond precisely to the equations for
a monoid object~\eqref{eq:monoid}. If we write $\mathbb{M}$ for the monoidal
category generated by the equational theory $\eqth{M}$, the algebras of
$\mathbb{M}$ in a strict monoidal category~$\mathcal{C}$ are precisely its
monoids: the category $\Alg{\mathbb{M}}{\mathcal{C}}$ of algebras of the
monoidal theory $\mathbb{M}$ in $\mathcal{C}$ is monoidally equivalent to the
category of monoids in~$\mathcal{C}$. Similarly, all the algebraic structures
introduced in this section can be defined using algebraic theories.

\begin{remark}
  The presentations given here are not necessarily minimal. For example, in the
  theory of commutative monoids one equation for units of
  monoids~\eqref{eq:monoid-string} is derivable from the
  equation~\eqref{eq:monoid_mult_comm}, one of the
  equations~\eqref{eq:monoid-nat} and one of the equations for units of
  monoids~\eqref{eq:monoid-string}:
  \[
  \sstrid{mult_unit_r}
  \!\qeq\!
  \sstrid{mult_unit_r_sym}
  \!\qeq\!
  \sstrid{mult_unit_l}
  \!\qeq\!
  \sstrid{mult_unit_c}
  \]
  A minimal presentation of this equational theory with three generators and
  seven equations is given in~\cite{massol:minimality}. However, not all the
  equational theories introduced in this paper have a known presentation which
  is proved to be minimal.
\end{remark}

\paragraph{Bialgebras.}
A \emph{bialgebra} $(B,\mu,\eta,\delta,\varepsilon,\gamma)$ in $\mathcal{C}$ is
an object $B$ together with five morphisms
\[
\begin{array}{r@{\ :\ }l@{\qquad}r@{\ :\ }l}
  \mu&B\otimes B\to B
  &
  \eta&I\to B\\
  \delta&B\to B\otimes B
  &
  \varepsilon&B\to I\\
\end{array}
\ \tand\ 
\gamma : B \otimes B\to B\otimes B
\]
such that $\gamma:B\otimes B\to B\otimes B$ is a symmetry for $B$,
$(B,\mu,\eta,\gamma)$ is a symmetric monoid and $(B,\delta,\varepsilon,\gamma)$
is a symmetric comonoid. The morphism $\gamma$ is thus pictured as
in~\eqref{eq:sym-string}, $\mu$ and $\eta$ as in \eqref{eq:monoid-string},
and~$\delta$ and~$\varepsilon$ as in~\eqref{eq:comonoid-string}. Those two
structures should be coherent, in the sense that the diagrams
\[
\begin{array}{c@{\hspace{-3ex}}c}
  \xymatrix{
    B\otimes B\ar[d]_{\delta\otimes\delta}\ar[r]^-{\mu}&B\ar[r]^-{\delta}&B\otimes B\\
    B\otimes B\otimes B\otimes B\ar[rr]_{B\otimes\gamma\otimes B}&&\ar[u]_{\mu\otimes\mu}B\otimes B\otimes B\otimes B\\
  }&
  \xymatrix{
    &B\ar[dr]^{\varepsilon}&\\
    I\ar[ur]^{\eta}\ar[rr]_{I}&&I
  }\\
  \xymatrix{
    &B\ar[dr]^{\varepsilon}&\\
    B\otimes B\ar[ur]^{\mu}\ar[rr]_{\varepsilon\otimes\varepsilon}&&I\otimes I=I
  }&
  \xymatrix{
    &B\ar[dr]^{\delta}&\\
    I=I\otimes I\ar[ur]^{\eta}\ar[rr]_{\eta\otimes\eta}&&B\otimes B
  }
\end{array}
\]
should commute. Graphically,
\begin{equation}
  \label{eq:bialg}
  \begin{array}{c}
    \sstrid{hopf_l}
    =
    \sstrid{hopf_r}
    \\
    \sstrid{counit_mult}
    =
    \sstrid{counit_x_counit}
    \qquad
    \sstrid{comult_unit}
    =
    \sstrid{unit_x_unit}
    \qquad
    \sstrid{unit_counit}=
  \end{array}
\end{equation}
should be satisfied.

A bialgebra is \emph{commutative} (\resp \emph{cocommutative}) when the induced
symmetric monoid $(B,\mu,\eta,\gamma)$ (\resp symmetric comonoid
$(B,\delta,\varepsilon,\gamma)$) is commutative (\resp cocommutative), and
\emph{bicommutative} when it is both commutative and cocommutative. A bialgebra
is \emph{qualitative}
when the diagram
\[
\xymatrix{
  &B\otimes B\ar[dr]^{\mu}&\\
  B\ar[ur]^{\delta}\ar[rr]_{B}&&B
}
\]
commutes. Graphically,
\begin{equation}
  \label{eq:qualitative}
  \strid{rel_l}
  \qeq
  \strid{rel_r}
\end{equation}



\begin{definition}
  We write $\eqth{B}$ for the \emph{equational theory of bicommutative
    bialgebras}. It has one generator for objects~$1$, five generators for
  morphisms
  \[
  \begin{array}{r@{\ :\ }l@{\qquad}r@{\ :\ }l}
    \mu&2\to 1
    &
    \eta&0\to 1\\
    \delta&1\to 2
    &
    \varepsilon&1\to 0\\
  \end{array}
  \ \tand\ 
  \gamma : 2\to 2
  \]
  and twenty-two relations: the two relations of symmetry~\eqref{eq:sym}, the
  eight relations of commutative monoids~\eqref{eq:monoid} \eqref{eq:monoid-nat}
  \eqref{eq:monoid_mult_comm}, the eight relations of cocommutative comonoids
  which are dual of~\eqref{eq:monoid} \eqref{eq:monoid-nat}
  \eqref{eq:monoid_mult_comm}, and the four compatibility relations for
  bialgebras~\eqref{eq:bialg}.

  We also write $\eqth{R}$ for the \emph{equational theory of qualitative
    bicommutative bialgebras} which is defined as~$\eqth{B}$, with the
  relation~\eqref{eq:qualitative} added.
\end{definition}

\paragraph{Dual objects.}
An object $L$ of $\mathcal{C}$ is said to be \emph{left dual} to an object $R$
when there exists two morphisms
\[
\eta:I\to R\otimes L
\qtand
\varepsilon:L\otimes R\to I
\]
called respectively the \emph{unit} and the \emph{counit} of the duality and
respectively pictured as
\[
\strid{adj_unit_lr}
\qtand
\strid{adj_counit_lr}
\]
making the diagrams
\[
\vxym{
  &L\otimes R\otimes L\ar[dr]^{\varepsilon\otimes L}&\\
  L\ar[ur]^{L\otimes\eta}\ar[rr]_L&&L\\
}
\qtand
\vxym{
  &R\otimes L\otimes R\ar[dr]^{R\otimes\varepsilon}&\\
  R\ar[ur]^{\eta\otimes R}\ar[rr]_{R}&&R\\
}
\]
commute. Graphically,
\[
\strid{zig_zag_l}
=
\strid{id_L}
\!\qtand\!
\strid{zig_zag_r}
=
\strid{id_R}
\]
We write~$\eqth{D}$ for the equational theory associated to dual objects.

\begin{remark}
  If $\mathcal{C}$ is a category, two dual objects in the monoidal category
  $\mathrm{End}(\mathcal{C})$ of endofunctors of $\mathcal{C}$, with tensor
  product given on objects by composition of functors, are adjoint endofunctors
  of $\mathcal{C}$. The theory of adjoint functors in a 2-category is described
  precisely in~\cite{street:free-adj}, the definition of $\eqth{D}$ is a
  specialization of this construction to the case where there is only one
  0-cell.
\end{remark}

\section{Presenting the category of relations}
\label{section:presentation-rel}
We now introduce a presentation of the category $\Rel$ of finite ordinals and
relations, by refining presentations of simpler categories. This result is
mentioned in Examples~6 and~7 of~\cite{hyland-power:symmetric-monoidal-sketches}
and is proved in three different ways
in~\cite{lafont:equational-reasoning-diagrams}, \cite{pirashvili:bialg-prop}
and~\cite{lack:composing-props}. The methodology adopted here to build this
presentation has the advantage of being simple to check (although very
repetitive) and can be extended to give the presentation of the category of
games and strategies described in Section~\ref{subsection:walking-inno}.

\paragraph{The simplicial category.}
The simplicial category $\Delta$ is the monoidal theory whose morphisms
$f:\intset{m}\to\intset{n}$ are the monotone functions from $\intset{m}$ to
$\intset{n}$. It has been known for a long time that this category is closely
related to the notion of monoid, see~\cite{maclane:cwm}
or~\cite{lafont:boolean-circuits} for example. This result can be formulated as
follows:
\begin{property}
  \label{property:delta-presentation}
  The monoidal category $\Delta$ is presented by the equational theory of
  monoids~$\eqth{M}$.
\end{property}
\noindent
In this sense, the simplicial category $\Delta$ impersonates the notion of
monoid. We extend here this result to more complex categories.




\paragraph{Multirelations.}
A \emph{multirelation} $R$ between two finite sets~$A$ and~$B$ is a function
\hbox{$R:A\times B\to\N$}. It can be equivalently be seen as a multiset whose
elements are in $A\times B$ or as a matrix over $\N$, or as a span
\[
\xymatrix@C=2ex@R=2ex{
  &\ar[dl]_sR\ar[dr]^t&\\
  A&&B\\
}
\]
in the category $\FinSet$ of finite sets -- for the latest case, the multiset
representation can be recovered from the span by
\[
R(a,b)\qeq\left|\setof{e\in R\tq s(e)=a\tand t(e)=b}\right|
\]
for every element $(a,b)\in A\times B$.
If \hbox{$R_1:A\to B$} and \hbox{$R_2:B\to C$} are two multirelations, their
composition is defined by
\[
R_2\circ R_1(a,c)
\qeq
\sum_{b\in B}R_1(a,b)\times R_2(b,c)
\tdot
\]
This corresponds to the usual composition of matrices if we see $R_1$ and $R_2$
as matrices over $\N$, and as the span obtained by computing the pullback
\[
\xymatrix@C=2ex@R=2ex{
  &&\ar[dl]R_2\circ R_1\ar[dr]&&\\
  &\ar[dl]_{s_1}R_1\ar[dr]^{t_1}&&\ar[dl]_{s_2}R_2\ar[dr]^{t_2}&\\
  A&&B&&C\\
}
\]
if we see $R_1$ and $R_2$ as spans in $\Set$. The cardinal $\card{R}$ of a
multirelation \hbox{$R:A\to B$} is the sum
\[
\card{R}\qeq\sum_{(a,b)\in A\times B}R(a,b)
\]
of its coefficients. We write $\FMR$ for the monoidal theory of multirelations:
its objects are finite ordinals and morphisms are multirelations between
them. It is a strict symmetric monoidal category with the tensor product
$\otimes$ defined on objects and morphisms by disjoint union, and thus a
monoidal theory.
In this category, the object~$\intset{1}$ can be equipped with the obvious
structure of bicommutative bialgebra
\begin{equation}
  \label{eq:mrel-bialg}
  (1,R^\mu,R^\eta,R^\delta,R^\varepsilon)
\end{equation}
In this structure, $R^\mu:\intset{2}\to\intset{1}$ is the multirelation defined
by $R^\mu(i,0)=1$ for $i=0$ or $i=1$, $R^\delta:\intset{1}\to\intset{2}$ is the
multirelation dual to~$R^\mu$, and $R^\eta:0\to 1$ and \hbox{$R^\varepsilon:1\to
  0$} are uniquely defined by the fact that the object $0$ is both initial and
terminal in~$\FMR$. We now show that the category of multirelations is presented
by the equational theory~$\eqth{B}$ of bicommutative bialgebras. We
write~$\mathcal{B}$ for the syntactic category of~$\eqth{B}$ (\ie the monoidal
category generated by the underlying signature of~$\eqth{B}$), so
that~$\catquot{\mathcal{B}}$ is the monoidal category generated by~$\eqth{B}$,
where~$\equiv$ is the congruence generated by the relations of~$\eqth{B}$. The
bicommutative bialgebra structure~\eqref{eq:mrel-bialg} induces an
``interpretation functor'' \hbox{$I:\mathcal{B}\to\FMR$} such that~$I(1)=1$,
$I(\mu)=R^\mu$, $I(\eta)=R^\eta$, $I(\delta)=R^\delta$ and
$I(\varepsilon)=R^\varepsilon$. Since, the morphisms~\eqref{eq:mrel-bialg}
satisfy the equations of bicommutative bialgebra, the interpretations of two
morphisms of~$\mathcal{B}$ related by~$\equiv$ will be equal. The interpretation
functor thus extends to a functor $\catquot{I}\ :\catquot{\mathcal{B}}\
\to\FMR$.

\begin{example}
  Consider the morphism
  \[
  ((\mu\otimes\eta\otimes 1)\circ(1\otimes\delta)\circ(\delta\otimes\varepsilon))\otimes 1
  \qcolon
  3\to 4
  \]
  of $\mathcal{B}$ whose graphical representation is
  \begin{equation}
    \label{ex:intp-sd}
    \strid{intp}
  \end{equation}
  Its interpretation is the multirelation
  \begin{equation}
    \label{ex:intp-mr}
    ((R^\mu\otimes R^\eta\otimes 1)\circ(1\otimes R^\delta)\circ(R^\delta\otimes
    R^\varepsilon))\otimes 1
  \end{equation}
  This multirelation is a function~$3\times 4\to\N$ (where $3$ and $4$ are
  respectively the sets~$\{0,1,2\}$ and~$\{0,1,2,3\}$) and can thus be
  represented as the following $\N$-valued matrix of size~\hbox{$3\times 4$}:
  \[
  \left(
    \begin{matrix}
      2&0&1&0\\
      0&0&0&0\\
      0&0&0&1\\
    \end{matrix}
  \right)
  \]
  This matrix is computed by evaluating the formula~\eqref{ex:intp-mr} but has
  in fact a very natural interpretation if we consider the string diagrammatic
  representation~\eqref{ex:intp-sd} of the morphism: an entry~$(i,j)$ of the
  matrix is precisely the number of different paths in wires linking the
  object~$i$ on the input to the object~$j$ on the output (for example, from~$0$
  there are two paths to~$0$ and one to~$2$, thus the first line of the matrix).
\end{example}


For every morphism $\phi:m+1\to n$ in $\mathcal{B}$, where $m>0$, we define a
morphism written \hbox{$S^{m\to n}\phi:m+1\to n$} by
\begin{equation}
  \label{eq:ctx-S}
  S^{m\to n}\phi\qeq \phi\circ(\gamma\otimes \id_{m-1})
\end{equation}
Graphically,
\[
S^{m\to n}\phi\qeq\strid{gsym_s}
\]
The \emph{stairs} morphisms are defined inductively as either $\id_1$ or
$S\phi'$ where $\phi'$ is a stair, and are represented graphically as
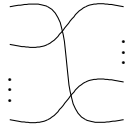
\[
\sstrid{gsym}
\]
The \emph{length} of a stairs is defined as $0$ if it is an identity~$\id_1$, or
as the length of the stairs~$\phi'$ plus one if it is of the form $S\phi'$. The
stairs of length~$n+1$ is written~$\gamma_n:n\to n$.

Morphisms~$\phi$ which are \emph{precanonical forms} are defined inductively:
$\phi$ is either empty or
\[
\begin{array}{ccccc}
  H^{m\to n}\phi'=\sstrid{bialg_nf_eta}
  &
  \tor
  &
  E^{m\to n}\phi'=\sstrid{bialg_nf_eps}
\end{array}
\]
or
\[
W_i^{m\to n}\phi'
\qeq
\sstrid{bialg_nf_mu}
\]
where $\phi:m\to n$ is a precanonical form. In this case, we write respectively
$\phi$ as~$Z:0\to 0$ (the identity morphism~$\id_{\intset{0}}$), as~$H^{m\to
  n}\phi':m\to n+1$, as \hbox{$E^{m\to n}\phi':m+1\to n$} or as $W_i^{m\to
  n}\phi':m\to n$ (where $i$ is the length of the stairs in the
morphism). Algebraically,
\[
Z=\id_0
\qquad
E^{m\to n}\phi'=\varepsilon\otimes\phi'
\qquad
H^{m\to n}\phi'=\eta\otimes\phi'
\]
and
\[
W_i^{m\to n}\phi'
\qeq
(i\otimes\mu\otimes(n-1-i))\circ(\gamma_i\otimes(n-i))\circ(1\otimes\phi')\circ(\delta\otimes(m-1))
\]
Precanonical forms~$\phi$ are thus the well formed morphisms (where compositions
respect types) generated by the following grammar:
\begin{equation}
  \label{eq:precan-mrel-gram}
  \phi\qgramdef Z\gramor H^{m\to n}\phi\gramor E^{m\to n}\phi\gramor W_i^{m\to n}\phi
\end{equation}
In order to simplify the notation, we will remove the superscripts in the
following and simply write~$W_i\phi$ instead of~$W_i^{m\to n}\phi$.

It is easy to remark that every non-identity morphism~$\phi$ of a category
generated by a monoidal equational theory (such as~$\eqth{B}$) can be written as
$\phi=(\intset{m}\otimes\pi\otimes\intset{n})\circ\phi'$, where~$\pi$ is a
generator, thus allowing us to reason inductively about morphisms, by case
analysis on the integer $\intset{m}$ and on the generator~$\pi$. Using this
technique, we can prove that
\begin{lemma}
  \label{lemma:mrel-precan}
  Every morphism~$\phi:m\to n$ of~$\mathcal{B}$ is equivalent (\wrt{} the relation
  $\equiv$) to a precanonical form.
\end{lemma}
\begin{proof}
  By induction on the size $\size{\phi}$ of~$\phi$.
  \begin{itemize}
  \item If~$\size{\phi}=0$ then $m=n$ and $\phi=\id_n$. If~$n=0$
    then~$\phi=Z$. Otherwise, we have $\phi=\id_{n+1}=1\otimes\id_n=W_0EH\id_n$
    and~$\id_n$ is equivalent to a canonical form by induction on~$n$.
  \item Otherwise, the morphism~$\phi$ is of the form~$\phi=\xi\circ\psi$
    with~$\size{\xi}=1$ and \hbox{$\size{\xi}+\size{\psi}=\size{\phi}$}. By
    induction hypothesis, the morphism~$\psi$ is equivalent to a canonical
    form. Moreover, the morphism~$\xi$ is of the form~$m_1\otimes\pi\otimes m_2$
    where~$\pi$ is either~$\mu$, $\eta$, $\delta$, $\varepsilon$ or~$\gamma$. We
    show the result by distinguishing these five cases for~$\pi$ and for each
    case by distinguishing whether the precanonical form of~$\psi$ is of the
    form~$Z$, $H\psi'$, $E\psi'$ or~$W_i\psi'$.
    \begin{enumerate}
    \item Suppose that~$\pi=\mu$.
      \begin{enumerate}
      \item If~$\psi=H\psi'$ then we distinguish two cases.
        \begin{itemize}
        \item If~$m_1=0$ then we have the equivalence
          \[
          \strid{bialg_mu_h_1_1}
          \qequiv
          \strid{bialg_mu_h_1_2}
          \]
          where~$\psi'$ is equivalent to a precanonical form by induction
          hypothesis.
        \item Otherwise, the morphism~$\phi$ can be represented by
          \[
          \strid{bialg_mu_h_2_1}
          \]
          and is of the form~$H(((m_1-1)\otimes\mu\otimes m_2)\circ\psi')$,
          where the morphism \hbox{$((m_1-1)\otimes\mu\otimes m_2)\circ\psi'$}
          is equivalent to a precanonical form by induction hypothesis.
        \end{itemize}
      \item If~$\psi=E\psi'$ then the morphism~$\phi$ can be represented by
        \[
        \strid{bialg_mu_e}
        \]
        and is of the form~$E(\xi\circ\psi')$ where the morphism~$\xi\circ\psi'$
        is equivalent to a precanonical form by induction hypothesis.
      \item If~$\psi=W_i'\psi'$ then we distinguish four cases
        \begin{itemize}
        \item If~$m_1<i-1$ then we have the equivalence
          \[
          \strid{bialg_mu_w_1_1}
          \qequiv
          \strid{bialg_mu_w_1_2}
          \]
          and~$\phi$ is of the form~$W_{i-1}(((m_1-1\otimes\mu\otimes
          m_2))\circ\psi')$ where the morphism $((m_1-1\otimes\mu\otimes
          m_2))\circ\psi'$ is equivalent to a precanonical form by induction
          hypothesis.
        \item If~$m_1=i-1$ then we have the equivalences
          \[
          \hspace{-2cm}
          \strid{bialg_mu_w_2_1}
          \equiv
          \strid{bialg_mu_w_2_2}
          \equiv
          \strid{bialg_mu_w_2_3}
          \]
          and we actually are in the case which is handled just below.
        \item If~$m_1=i$ then we have the equivalence
          \[
          \strid{bialg_mu_w_3_1}
          \qequiv
          \strid{bialg_mu_w_3_2}
          \]
          and~$\phi$ is of the form~$W_i(\xi\circ\psi')$ where the morphism
          $\xi\circ\psi'$ is equivalent to a precanonical form by induction
          hypothesis.
        \item If~$m_1>i$ then~$\phi$ can be represented by
          \[
          \strid{bialg_mu_w_4_1}
          \]
          and is of the form~$W_i(\xi\circ\psi')$ where the
          morphism~$\xi\circ\psi'$ is equivalent to a precanonical form by
          induction hypothesis.
        \end{itemize}
      \end{enumerate}
    \item Suppose that~$\pi=\eta$.
      \begin{enumerate}
      \item If~$\psi=Z$ then~$\phi=HZ$ which is a precanonical form.
      \item If~$\psi=H\psi'$ then we distinguish two cases.
        \begin{itemize}
        \item If~$m_1=0$ then~$\phi=HH\psi'$ which is a precanonical form.
        \item Otherwise, $\phi=H(((m_1-1)\otimes\eta\otimes m_2)\circ\psi')$
          where \hbox{$(m_1-1)\otimes\eta\otimes m_2$} is equivalent to a
          precanonical form by induction hypothesis.
        \end{itemize}
      \item If~$\psi=E\psi'$ then~$\phi=E(\xi\circ\psi')$ where the
        morphism~$\xi\circ\psi'$ is equivalent to a precanonical form by
        induction hypothesis.
      \item If~$\psi=W_i\psi'$ then we distinguish two cases.
        \begin{itemize}
        \item If~$m_1\leq i$ then~$\phi\equiv W_{i+1}(\xi\circ\psi')$ where the
          morphism~$\xi\circ\psi'$ is equivalent to a precanonical form by
          induction hypothesis.
        \item Otherwise, $\phi=W_i(\xi\circ\psi')$ where the
          morphism~$\xi\circ\psi'$ is equivalent to a precanonical form by
          induction hypothesis.
        \end{itemize}
      \end{enumerate}
    \item Suppose that $\pi=\delta$.
      \begin{enumerate}
      \item If $\psi=H\psi'$ then we distinguish two cases.
        \begin{itemize}
        \item If $m_1=0$ then $\phi\equiv HH\psi'$ where $\psi'$ is a
          precanonical form.
        \item Otherwise, $\phi\equiv H(((m_1-1)\otimes\delta\otimes
          m_2)\circ\psi')$ where $((m_1-1)\otimes\delta\otimes m_2)\circ\psi'$
          is equivalent to a precanonical form by induction hypothesis.
        \end{itemize}
      \item If $\psi=E\psi'$ then $\phi=E(\xi\circ\psi')$ where the morphism
        $\xi\circ\psi'$ is equivalent to a precanonical form by induction
        hypothesis
      \item If $\psi=W_i\psi'$ the we distinguish three cases.
        \begin{itemize}
        \item If $m_1<i$ then $\phi\equiv W_{i+1}(\xi\circ\psi')$ where the
          morphism $\xi\circ\psi'$ is equivalent to a precanonical form by
          induction hypothesis
        \item If $m_1=i$ then $\phi\equiv W_iW_{i+1}(\xi\circ\psi')$ where the
          morphism $\xi\circ\psi'$ is equivalent to a precanonical form by
          induction hypothesis.
        \item Otherwise, $\phi=W_i(\xi\circ\psi')$ where the morphism
          $\xi\circ\psi'$ is equivalent to a precanonical form by induction
          hypothesis.
        \end{itemize}
      \end{enumerate}
    \item Suppose that $\pi=\varepsilon$.
      \begin{enumerate}
      \item If $\psi=H\psi'$ then we distinguish two cases.
        \begin{itemize}
        \item If $m_1=0$ then $\phi\equiv\psi'$ where the morphism $\psi'$ is a
          precanonical form.
        \item Otherwise, $\psi=H(((m_1-1)\otimes\varepsilon\otimes
          m_2)\circ\psi')$ where $((m_1-1)\otimes\varepsilon\otimes
          m_2)\circ\psi'$ is equivalent to a precanonical form by induction
          hypothesis.
        \end{itemize}
      \item If $\psi=E\psi'$ then $\phi=E(\xi\circ\psi')$ where the morphism
        $\xi\circ\psi'$ is equivalent to a precanonical form by induction
        hypothesis.
      \item If $\psi=W_i\psi'$ then we distinguish three cases.
        \begin{itemize}
        \item If $m_1<i$ then $\phi\equiv W_{i-1}(\xi\circ\psi')$ where the
          morphism $\xi\circ\psi'$ is equivalent to a precanonical form by
          induction hypothesis.
        \item If $m_1=i$ then $\phi\equiv E(\xi\circ\psi')$ where the morphism
          $\xi\circ\psi'$ is equivalent to a precanonical form by induction
          hypothesis.
        \item Otherwise, $\phi=W_i(\xi\circ\psi')$ where the morphism
          $\xi\circ\psi'$ is equivalent to a precanonical form by induction
          hypothesis.
        \end{itemize}
      \end{enumerate}
    \item Suppose that $\pi=\gamma$.
      \begin{enumerate}
      \item If $\psi=H\psi'$ then we distinguish two cases.
        \begin{itemize}
        \item If $m_1=0$ then $\phi\equiv((1\otimes\eta\otimes m_2)\circ\psi')$
          where the morphism $(1\otimes\eta\otimes m_2)\circ\psi'$ is equivalent
          to a precanonical form by induction hypothesis.
        \item Otherwise, $\phi=H(\xi\circ\psi')$ where the morphism
          $\xi\circ\psi'$ is equivalent to a precanonical form by induction
          hypothesis.
        \end{itemize}
      \item If $\psi=E\psi'$ then $\phi=E(\xi\circ\psi')$ where the morphism
        $\xi\circ\psi'$ is equivalent to a precanonical form by induction
        hypothesis.
      \item If $\psi=W_i\psi'$ then we distinguish four cases.
        \begin{itemize}
        \item If $m_1<i-1$ then $\phi\equiv W_i(\xi\circ\psi')$ where the
          morphism $\xi\circ\psi'$ is equivalent to a precanonical form by
          induction hypothesis.
        \item If $m_1=i-1$ then $\phi\equiv W_{i-1}(((m_1+1)\otimes\gamma\otimes
          m_2)\circ\psi')$ where the morphism $((m_1+1)\otimes\gamma\otimes
          m_2)\circ\psi'$ is equivalent to a precanonical form by induction
          hypothesis.
        \item If $m_1=i$ then $\phi\equiv W_{i+1}(((m_1+1)\otimes\gamma\otimes
          m_2)\circ\psi')$ where the morphism $((m_1+1)\otimes\gamma\otimes
          m_2)\circ\psi'$ is equivalent to a precanonical form by induction
          hypothesis.
        \item Otherwise, $\phi=W_i(\xi\circ\psi')$ where the morphism
          $\xi\circ\psi'$ is equivalent to a precanonical form by induction
          hypothesis.\qedhere
        \end{itemize}
      \end{enumerate}
    \end{enumerate}
  \end{itemize}
\end{proof}

The \emph{canonical forms} are precanonical forms which are normal \wrt{} the
following rewriting system:
\begin{equation}
  \label{eq:mrel-cf-rs}
  \begin{array}{r@{\quad\Longrightarrow\quad}l}
    HW_i&W_{i+1}H\\
    HE&EH\\
    W_iW_j&W_jW_i\qquad\text{when $i<j$}
  \end{array}
\end{equation}
when considered as words generated by the grammar~\eqref{eq:precan-mrel-gram}.
It is routine verifications to show that two precanonical forms~$\phi$
and~$\psi$ such that~$\phi$ rewrites to~$\psi$ are equivalent. This rewriting
system thus provides us with a notion of canonical form for precanonical forms:

\begin{lemma}
  \label{lemma:mrel-can}
  The rewriting system~\eqref{eq:mrel-cf-rs} is normalizing.
\end{lemma}
\begin{proof}
  We first show that the rewriting system is terminating by defining an
  interpretation of precanonical forms into~$\N\times\N$, ordered
  lexicographically. This interpretation~$\intp{-}$ is defined on generators by
  \[
  \intp{Z}=(0,0)
  \qquad
  \intp{H}=(0,0)
  \qquad
  \intp{E}=(1,0)
  \qquad
  \intp{W_i}=(1,i)
  \]
  and on composition and identities by
  \[
  \intp{G\circ F}=(\intp{G}_1+2\times\intp{F_1},\intp{G}_2+2\times\intp{F}_2)
  \qtand
  \intp{\id}=(0,0)
  \]
  where~$F$ and~$G$ are such that~$\intp{F}=(\intp{F}_1,\intp{F}_2)$
  and~$\intp{G}=(\intp{G}_1,\intp{G}_2)$. It can be remarked that the rules are
  strictly decreasing \wrt this interpretation:
  \[
  \intp{HW_i}=(2,2i)>(1,i)=\intp{W_iH}
  \qquad
  \intp{HE}=(2,0)>(1,0)=\intp{EH}
  \]
  and
  \[
  \intp{W_iW_j}=(3,i+2j)>(3,j+2i)=\intp{W_jW_i}
  \]
  The rewriting system is therefore terminating. It moreover locally confluent,
  since the two critical pairs are joinable:
  
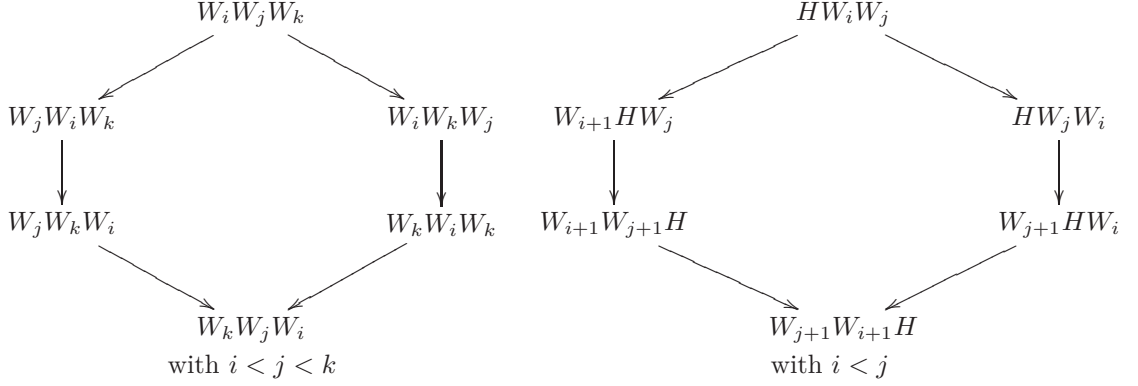
\[
  \hspace{-6mm}
  \begin{array}{cc}
    \vxym{
      &\ar[dl]W_iW_jW_k\ar[dr]&\\
      W_jW_iW_k\ar[d]&&W_iW_kW_j\ar[d]\\
      W_jW_kW_i\ar[dr]&&W_kW_iW_k\ar[dl]\\
      &W_kW_jW_i&\\
    }
    &
    \vxym{
      &\ar[dl]HW_iW_j\ar[dr]&\\
      W_{i+1}HW_j\ar[d]&&HW_jW_i\ar[d]\\
      W_{i+1}W_{j+1}H\ar[dr]&&W_{j+1}HW_i\ar[dl]\\
      &W_{j+1}W_{i+1}H&\\
    }
    \\
    \text{with $i<j<k$}
    &
    \text{with $i<j$}
  \end{array}
  \]
  The rewriting system being terminating, it is thus confluent.
\end{proof}
\begin{remark}
  Canonical forms are the precanonical forms of the form
  \begin{equation}
    \label{eq:mrel-can}
    W_{i^n_{k_n}}\cdots W_{i^n_1}E\cdots\cdots W_{i^1_{k_1}}\cdots W_{i^1_1}EH\cdots HZ
  \end{equation}
  with $i^p_1\geq\ldots\geq i^p_{k_p}$, for every~$p$ such that $1\leq p\leq n$.
\end{remark}

From Lemmas~\ref{lemma:mrel-precan} and~\ref{lemma:mrel-can}, we can finally
deduce that every morphism of the category~$\mathcal{B}$ is equivalent to an
unique canonical form.

\begin{lemma}
  \label{lemma:mrel-bij}
  The interpretation functor~$\catquot{I}\ :\catquot{\mathcal{B}}\ \to\FMR$ is full.
\end{lemma}
\begin{proof}
  We show the result by showing that the functor~$I:\mathcal{B}\to\FMR$ is full,
  \ie that every multirelation~$R:\intset{m}\to\intset{n}$ is the image of a
  precanonical form~$\phi:m\to n$ in~$\mathcal{B}$, by induction on~$m$ and on
  the cardinal $\card{R}$ of $R$.
  \begin{enumerate}
  \item If $m=0$ then $R$ is the interpretation of the precanonical form
    $H\ldots HZ$, with~$n$ occurrences of $H$.
  \item If $m>0$ and for every $j<n$, $R(0,j)=0$ then $R$ is of the form
    $R=R^\varepsilon\otimes R'$, where~$R':\intset{m-1}\to\intset{n}$ is the
    multirelation such that~$R'(i,j)=R(i+1,j)$. By induction hypothesis, $R'$ is
    the interpretation of a precanonical form~$\phi'$ and~$R$ is therefore the
    interpretation of the precanonical form~$E\phi'$.
  \item Otherwise, we necessarily have~$n\neq 0$ and there exists and index~$k'$
    such that \hbox{$R(0,k)\neq 0$}. We write~$k$ for the greatest such
    index. The multirelation~$R$ is of the form
    \[
    R
    \qeq
    (\intset{k}\otimes R^\mu\otimes\intset{n-1-k})\circ(R^{\gamma_k}\otimes\intset{n-k})\circ(1\otimes R')\circ(R^\delta\otimes\intset{m-1})
    \]
    Where $R':\intset{m}\to\intset{n}$ is the multirelation defined by
    $R'(0,k)=R(0,k)-1$ and \hbox{$R'(i,j)=R(i,j)$} for
    every~$(i,j)\neq(0,k)$. The multirelation~$R'$ is thus of cardinal
    \hbox{$\card{R'}=\card{R}-1$} and is the interpretation of a precanonical
    form~$\phi':m\to n$ by induction hypothesis. Finally, $R$ is the
    interpretation of the precanonical form~$W_k\phi'$.\qedhere
  \end{enumerate}
\end{proof}

The proof of the previous lemma provides us with an algorithm which, given a
multirelation~$R$, builds a precanonical form~$\phi$ whose interpretation
is~$R$. The execution of this algorithm consists in enumerating the coefficients
of the multirelation column after column. We suppose given a
multirelation~$R:\intset{m}\to\intset{n}$. In pseudo-code, the algorithm can be
written as follows:

\vspace{2ex}

\noindent
for $i=0$ to $m-1$ do\\
\null\qquad for $j=n-1$ downto $0$ do\\
\null\qquad\qquad for $k=0$ to $R(i,j)$ do\\
\null\qquad\qquad\qquad print ``$W_j$''\\
\null\qquad\qquad done\\
\null\qquad\qquad print ``$H$''\\
\null\qquad done\\
done\\
for $j=0$ to $n-1$ do\\
\null\qquad print ``$E$''\\
done\\
print ``$Z$''

\vspace{2ex}

\noindent
The word printed by the algorithm will be a precanonical form whose
interpretation is~$R$.

Knowing the general form~\eqref{eq:mrel-can} of canonical forms, it is easy to
show that the precanonical form produced by the algorithm are actually canonical
forms. Conversely, every canonical form~\eqref{eq:mrel-can} can be read as an
``enumeration'' of the coefficients of a multirelation in a way similar the
previous algorithm. This shows that, in fact,
multirelations~$R:\intset{m}\to\intset{n}$ are in bijection with the canonical
forms~$\phi:m\to n$. A morphism of~$\mathcal{B}$ being equivalent to an unique
canonical form, we finally deduce that

\begin{theorem}
  \label{thm:fmr-pres}
  The categories~$\catquot{\mathcal{B}}$ and~$\FMR$ are isomorphic, \ie the
  category~$\FMR$ of natural numbers and multirelations is presented by the
  theory~$\eqth{B}$ of bicommutative bialgebras.
\end{theorem}

\paragraph{Relations.}
The monoidal category $\Rel$ has finite ordinals as objects and relations as
morphisms. This category can be obtained from $\FMR$ by quotienting the
morphisms by the equivalence relation $\sim$ on multirelations such that two
multirelations $R_1,R_2:m\to n$ are equivalent when they have the same null
coefficients. We can therefore easily adapt the previous presentation to show
that

\begin{theorem}
  The category $\Rel$ of relations is presented by the equational theory
  $\eqth{R}$ of \emph{qualitative} bicommutative bialgebras.
\end{theorem}

\noindent
In particular, precanonical forms are the same and canonical forms are defined
by adding the rule
\begin{equation}
  \label{eq:rel-cf-rs}
  W_iW_i\Longrightarrow W_i
\end{equation}
to the rewriting system~\eqref{eq:mrel-cf-rs}, which remains normalizing.

\section{A game semantics for first-order causality}
\label{section:games-strategies}

Suppose that we are given a fixed first-order language~$\mathcal{L}$, that is
\begin{itemize}
\item a set of proposition symbols~$P,Q,\ldots$ with given arities,
\item a set of function symbols~$f,g,\ldots$ with given arities,
\item and a set of first-order variables~$x,y,\ldots$.
\end{itemize}
\emph{Terms}~$t$ and \emph{formulas}~$A$ are respectively generated by the
following grammars:
\[
\begin{array}{rcl}
  t&\qgramdef&x\gramor f(t,\ldots,t)
  \\
  A&\qgramdef&P(t,\ldots,t)\gramor\qforall{x}{A}\gramor\qexists{x}{A}
\end{array}
\]
(we only consider formulas without connectives here). We suppose that
application of propositions and functions always respect arities. Formulas are
considered modulo renaming of bound variables and substitution $A[t/x]$ of a
free variable $x$ by a term $t$ in a formula $A$ is defined as usual, avoiding
capture of variables. In the following, we sometimes omit the arguments of
propositions when they are clear from the context. We also suppose given a
set~$\axioms$ of \emph{axioms}, that is pairs of propositions, which is
reflexive, transitive and closed under substitution (so that the obtained logic
has the cut-elimination property). The logic associated to these formulas has
the following inference rules:
\[
\begin{array}{c@{\qquad}c}
  \inferrule{A[t/x]\vdash B}{\qforall x A\vdash B}{\lrule{$\forall$-L}}
  &
  \inferrule{A\vdash B}{A\vdash \qforall x B}{\lrule{$\forall$-R}}
  \displaybreak[0]
  \\
  &
  \text{(with $x$ not free in $A$)}
  \displaybreak[0]
  \\[2ex]
  \inferrule{A\vdash B}{\qexists x A\vdash B}{\lrule{$\exists$-L}}
  &
  \inferrule{A\vdash B[t/x]}{A\vdash \qexists x B}{\lrule{$\exists$-R}}
  \displaybreak[0]
  \\
  \text{(with $x$ not free in $B$)}
  &
  \displaybreak[0]
  \\[2ex]
  \inferrule{(P,Q)\in\axioms}{P\vdash Q}{\lrule{Ax}}
  &
  \inferrule{A\vdash B\\B\vdash C}{A\vdash C}{\lrule{Cut}}
\end{array}
\]

\paragraph{Games and strategies.}
\label{subsection:games-strategies}
Games are defined as follows.

\begin{definition}
  A \emph{game} $A=(\moves{A},\lambda_A,\leq_A)$ consists of a set~$\moves{A}$
  whose elements are called \emph{moves}, a function $\lambda_A$
  from~$\moves{A}$ to~$\{-1,+1\}$ which to every move $m$ associates its
  \emph{polarity}, and a partial order $\leq_A$ on moves, called
  \emph{causality} or \emph{justification}, which should be well-founded, \ie
  such that every move $m\in\moves{A}$ defines a finite downward closed set
  \[
  m\!\downarrow\qeq\setof{n\in\moves{A}\tq n\leq_A m}
  \]
  A move $m$ is said to be a \emph{Proponent move} when $\lambda_A(m)=+1$ and an
  \emph{Opponent move} otherwise.
\end{definition}
\noindent
The size~$\size{A}$ of a game~$A$ is the cardinal of its set of
moves~$\moves{A}$.

\begin{remark}
  More generally, games should be defined as event
  structures~\cite{winskel:event-structures} in order to be able to model
  additive connectives. We don't detail this here since we only consider
  formulas without connectives.
\end{remark}

If $A$ and $B$ are two games, their tensor product $A\otimes B$ is defined by
disjoint union on moves, polarities and causality:
\[
\moves{A\otimes B}=\moves{A}\uplus\moves{B}
\tcomma\quad
\lambda_{A\otimes B}=\lambda_A+\lambda_B
\qtand
\leq_{A\otimes B}=\leq_A\cup\leq_B
\]
The opposite game $A^*$ of the game $A$ is obtained from~$A$ by inverting
polarities of moves:
\[
A^*=(\moves{A},-\lambda_A,\leq_A)
\tdot
\]
Finally, the arrow game $A\llimp B$ is defined by
\[
A\llimp B\qeq A^*\otimes B
\tdot
\]
A game $A$ is \emph{filiform} when the associated partial order is total (we are
mostly interested in such games in the following).


\begin{definition}
  \label{def:strategy}
  A \emph{strategy} $\sigma$ on a game $A$ is a partial order $\leq_\sigma$ on
  the moves of $A$ which satisfies the two following properties:
  \begin{enumerate}
  \item \emph{polarity}: for every pair of moves $m,n\in\moves{A}$,
    \[
    m<_\sigma n
    \qqtimpl
    \lambda_A(m)=-1
    \qtand
    \lambda_A(n)=+1
    \]
  \item \emph{acyclicity}: the partial order $\leq_\sigma$ is compatible with
    the partial order of the game, in the sense that the transitive closure of
    their union is still a partial order (\ie is acyclic).
  \end{enumerate}
\end{definition}
The \emph{size} $\size{A}$ of a game $A$ is the cardinal of $\moves{A}$ and the
\emph{size} $\size{\sigma}$ of a strategy $\sigma:A$ is the cardinal of the
relation~$\leq_\sigma$.

\paragraph{A category of games.}
At this point it would be very tempting to build a category whose
\begin{itemize}
\item objects are games,
\item morphisms~$\sigma:A\to B$ are strategies on the game~$A\llimp B$.
\end{itemize}
The identity strategy~$\id_A:A'\to A$ (the apostrophe sign is only used here to
identify unambiguously the two copies of~$A$) would be the strategy such that
for every move~$m$ in~$A$ and~$m'$ in~$A'$, which are instances of a same
move~$m$, we have~\hbox{$m'\leq_{\id_A}m$} whenever~$\lambda_A(m)=+1$
and~$m\leq_{\id_A}m'$ whenever~$\lambda_A(m)=-1$ (it can easily be checked that
this definition satisfies the axioms for strategies). Now consider two
strategies~$\sigma:A\to B$ and~$\tau:B\to C$. The partial order~$\leq_\sigma$ on
the set~$\moves{A}\uplus\moves{B}$ is relation on~$\moves{A}\uplus\moves{B}$,
\ie a subset of~$(\moves{A}\uplus\moves{B})^2$, and similarly for~$\tau$. The
partial order~$\leq_{\tau\circ\sigma}$ corresponding to composite
$\tau\circ\sigma:A\to C$ of the two strategies~$\sigma$ and~$\tau$ would be
defined as the transitive closure of the relation~$\leq_\sigma\cup\leq_\tau$
on~$\moves{A}\uplus\moves{B}\uplus\moves{C}$ restricted to the
set~$\moves{A}\uplus\moves{C}$. It is easily checked that identities act as
neutral elements for composition. Similar ideas for composing strategies were in
particular developed in the appendix of~\cite{hyland-schalk:games-graphs}.

For example, consider the game~$A$ with two Proponent moves~$m_1$ and~$m_2$ and
the empty causality relation, the game~$B$ with two Proponent moves~$n_1$
and~$n_2$ and the causality relation~$n_1\leq_B n_2$, the strategy~$\sigma:A'\to
A$ such that~$m_1'\leq_\sigma m_2$ and~$m_2'\leq_\sigma m_1$ and the
strategy~$\tau:A\to B$ such that~$m_1\leq_\tau n_1$ and~$m_2\leq_\tau
n_2$. Their composite is the strategy~$\tau\circ\sigma:A'\to B$ such
that~$m_2'\leq_{\tau\circ\sigma}n_1$ and $m_1'\leq_{\tau\circ\sigma}n_2$. This
can be viewed graphically as follows:
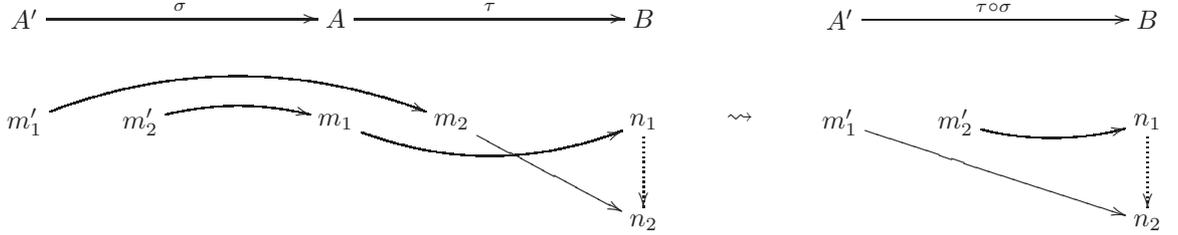
\[
\hspace{-4ex}
\vxym{
  A'\ar[rrr]^\sigma&&&A\ar[rrr]^\tau&&&B\\
  m_1'\ar@/^4ex/[rrrr]&m_2'\ar@/^/[rr]&&m_1\ar@/_3ex/[rrr]&m_2\ar[drr]&&n_1\ar@{.>}[d]\\
  &&&&&&n_2\\
}
\qquad\rightsquigarrow\qquad
\vxym{
  A'\ar[rrr]^{\tau\circ\sigma}&&&B\\
  m_1'\ar[drrr]&m_2'\ar@/_/[rr]&&n_1\ar@{.>}[d]\\
  &&&n_2\\
}
\]
In the diagram above the dotted arrows represent the causal dependencies in the
games and solid arrows the dependencies in the strategies.

However, the composite of two strategies is not necessarily a strategy! For
example consider the game~$A$ defined as before excepted that~$m_1$ is now an
Opponent move, the game~$B$ defined as before excepted that~$n_2$ is now an
Opponent move, the strategy \hbox{$\sigma:0\to A$} (where $0$ denotes the empty
game) such that~$m_1\leq_\sigma m_2$ and the strategy \hbox{$\tau:A\to B$} such
that~$n_2\leq_\tau m_1$ and~$m_2\leq_\tau n_1$. Their ``composite'' is
\emph{not} a strategy because it does not satisfy the acyclicity property:
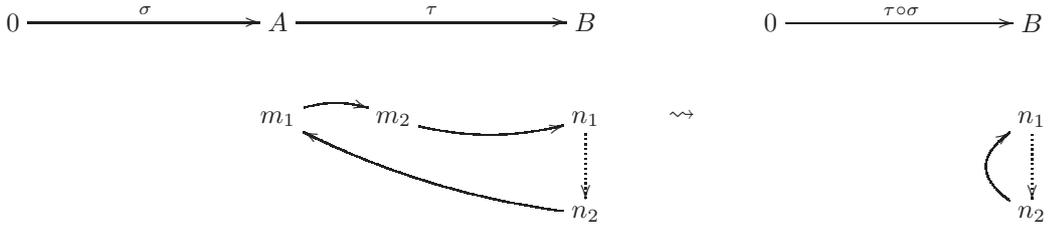
\[
\hspace{-1ex}
\vxym{
  0\ar[rrr]^\sigma&&&A\ar[rrr]^\tau&&&B\\
  &&&m_1\ar@/^/[r]&m_2\ar@/_/[rr]&&n_1\ar@{.>}[d]\\
  &&&&&&\ar@/^/[ulll]n_2\\
}
\qquad\rightsquigarrow\qquad
\vxym{
  0\ar[rrr]^{\tau\circ\sigma}&&&B\\
  &&&n_1\ar@{.>}[d]\\
  &&&\ar@/^4ex/[u]n_2\\
}
\]
This is a typical example of the fact that compositionality of strategies in
game semantics is often a subtle property that should be checked very carefully.

\begin{remark}
  \label{rem:ll-acyclicity}
  A more conceptual explanation of this compositionality problem can be given as
  follows. If we write~$P$ for the game with only one Proponent move, the
  game~$A$ should correspond, in a model of linear logic to either the tensor or
  the par of $P$ and $P^*$. However, we have not included in our strategies
  conditions which are necessary to distinguish between the interpretation of
  tensor and par. This explains why we are not able to recover the
  compositionality of the acyclicity property, which is deeply linked with the
  correctness criterion of linear logic. We leave a precise investigation of
  this for future works, in which we plan to extend our model to first-order
  linear logic.
\end{remark}

Fortunately, if we restrict the previous attempt of construction of a category,
by only allowing \emph{finite filiform games} as objects, then we actually
construct a category (\ie the composite of two morphisms is a morphism) that we
write~$\Games$. Moreover, we show that the connective-free fragment of
first-order propositional logic can be interpreted in this category and that the
conditions imposed on strategies characterize exactly the strategies
interpreting proofs (Theorem~\ref{thm:definability}).

We could give a direct proof of the fact that~$\Games$ is actually a
category. However, a direct proof of the fact that the composite of two acyclic
strategies is acyclic is combinatorial, lengthy and requires global reasoning
about strategies. This proof would show, by reductio ad absurdum, that if the
composite of two strategies contains a cycle (together with the causality of the
game) then one of the strategies already contains a cycle. So, it would moreover
not be very satisfactory in the sense that it would not be constructive. Instead
of proceeding in this way, we define the category~$\Games$ in an abstract
fashion, construct a presentation of this category, and conclude \emph{a
  posteriori} that in fact its only morphisms are strategies, which implies in
particular (Theorem~\ref{thm:composition}) that strategies do actually compose!

We first define a weaker notion of strategy
\begin{definition}
  A \emph{cyclic strategy}~$\sigma$ on a game~$A$ is a relation on the moves
  of~$A$, \ie a subset of~$\moves{A}\times\moves{A}$, such that
  \begin{enumerate}
  \item the relation~$\sigma$ is reflexive and transitive,
  \item \emph{polarity}: for every pair of moves $m,n\in\moves{A}$,
    \[
    m\mathop{\sigma}n
    \qtand
    m\neq n
    \qqtimpl
    \lambda_A(m)=-1
    \qtand
    \lambda_A(n)=+1
    \]
  \end{enumerate}
\end{definition}
In particular, every strategy is a cyclic strategy. From this definition it is
very easy to build a category~$\CGames$ whose
\begin{itemize}
\item objects are games,
\item morphisms~$\sigma:A\to B$ are strategies on the game~$A\llimp B$,
\item identities and composition are defined as above.
\end{itemize}
Since the definition of cyclic strategy is much weaker than the notion of
strategy, it is routine to check that the category is well-defined. We now
define the category~$\Games$ as the category generated in~$\CGames$ by finite
filiform games and strategies, \ie the smallest category whose
\begin{itemize}
\item objects are finite filiform games,
\item for every objects~$A$ and~$B$, and every strategy~$\sigma:A\llimp B$ in
  the sense of Definition~\ref{def:strategy}, we have that $\sigma$ is a
  morphism in~$\Hom(A,B)$,
\item for every objects~$A$, $B$ and~$C$, if~$\sigma$ is a morphism in
  $\Hom(A,B)$ and~$\tau$ is a morphism in~$\Hom(B,C)$ then their
  composite~$\tau\circ\sigma$ (in the category~$\CGames$) is a morphism
  in~$\Hom(A,C)$.
\end{itemize}
As mentioned above, we will show in Theorem~\ref{thm:composition} that the only
morphisms of this category are actually strategies.

\paragraph{A monoidal structure on $\Games$.}
If $A$ and $B$ are two games, the game $A\before{}B$ (to be read $A$
\emph{before} $B$) is the game defined as $A\lltens B$ on moves and polarities
and $\leq_{A\before{}B}$ is the transitive closure of the relation
\[
\leq_{A\lltens B}\cup\;\setof{(a,b)\tq a\in \moves{A}\tand b\in\moves{B}}
\]
This operation is extended as a bifunctor on strategies as follows. If
$\sigma:A\to B$ and $\tau:C\to D$ are two strategies, the strategy
$\sigma\before{}\tau:A\before{}C\to B\before{}D$ is defined as the relation
$\leq_{\sigma\before{}\tau}=\leq_\sigma\uplus\leq_\tau$.
This bifunctor induces a monoidal structure $(\Games,\before{},I)$ on the
category $\Games$, where $I$ denotes the empty game.

We write $O$ for a game with only one Opponent move and $P$ for a game with only
one Proponent move. It can be easily remarked that finite filiform games $A$ are
generated by the following grammar
\[
A\qqgramdef I\gramor O\before{}A\gramor P\before{}A
\]
A game $X_1\before{}\cdots\before{}X_n\before{}I$ where the $X_i$ are either $O$
or $P$ is represented graphically as
\[
\begin{array}{c}
  X_1\\
  \vdots\\
  X_n
\end{array}
\]
and a strategy $\sigma:A\to B$ is represented graphically by drawing a line from
a move $m$ to a move $n$ whenever $m\leq_\sigma n$. For example, the strategy
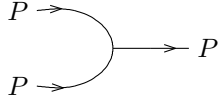
\hbox{$\mu^P:P\before{}P\to P$}
\[
\strid{mult_P}
\]
is the strategy on the game $(O\before{}O)\otimes P$ in which both Opponent move
of the left-hand game justify the Proponent move of the right-hand game. When a
move does not justify (or is not justified by) any other move, we draw a line
ended by a small circle. For example, the strategy \hbox{$\varepsilon^P:P\to
  I$}, drawn as
\[
\strid{counit_P_small}
\]
is the unique strategy from $P$ to the terminal object $I$. With these
conventions, we introduce notations for some morphisms which are depicted in
Figure~\ref{fig:inno-gen}.

\begin{figure}[htp!]
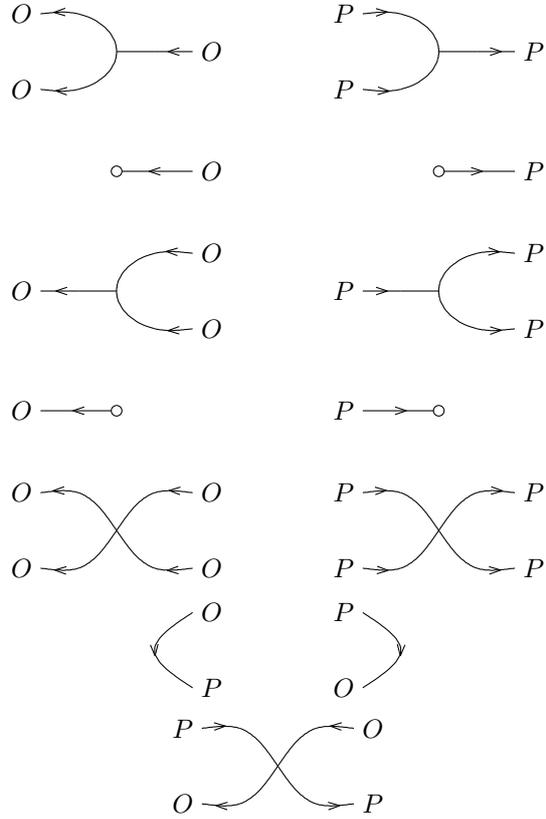

    \vbox{
      \[
        \begin{array}{c}
          \begin{array}{r@{\qcolon}l@{\quad}r@{\qcolon}l}
            \mu^O&O\before{} O\to O&\mu^P&P\before{} P\to P\\
            \eta^O&I\to O&\eta^P&I\to P\\
            \delta^O&O\to O\before{} O&\delta^P&P\to P\before{} P\\
            \varepsilon^O&O\to I&\varepsilon^P&P\to I\\
            \gamma^O&O\before{} O\to O\before{} O&\gamma^P&P\before{} P\to P\before{} P\\
            \eta^{OP}&I\to O\before{} P&\varepsilon^{OP}&P\before{} O\to I\\
          \end{array}
          \\
          \begin{array}{rcl}
            \gamma^{OP}&\colon&P\before{} O\to O\before{} P
          \end{array}
        \end{array}
      \]
      respectively drawn as
      \[
        \begin{array}{c}
          \begin{array}{c@{\qquad}c}
            \strid{mult_O}&\strid{mult_P}\\[4ex]
            \strid{unit_O}&\strid{unit_P}\\[4ex]
            \strid{comult_O}&\strid{comult_P}\\[4ex]
            \strid{counit_O}&\strid{counit_P}\\[4ex]
            \strid{sym_O}&\strid{sym_P}\\[4ex]
            \strid{unit_OP}&\strid{counit_OP}
          \end{array}
          \\
          \begin{array}{c}
            \strid{sym_OP}
          \end{array}
        \end{array}
      \]
    }
  \caption{Generators of the strategies.}
  \label{fig:inno-gen}
\end{figure}

\paragraph{A game semantics.}
A formula $A$ is interpreted as a filiform game~$\intp{A}$ by
\[
\intp{P}=I
\qquad
\intp{\qforall x A}=O\before{}\intp{A}
\qquad
\intp{\qexists x A}=P\before{}\intp{A}
\]
A cut-free proof $\pi:A\vdash B$ is interpreted as a strategy
$\sigma:\intp{A}\llimp\intp{B}$ whose causality partial order $\leq_\sigma$ is
defined as follows. For every Proponent move $P$ interpreting a quantifier
introduced by a rule which is either
\[
\inferrule{A[t/x]\vdash B}{\qforall x A\vdash B}{\lrule{$\forall$-L}}
\qqtor
\inferrule{A\vdash B[t/x]}{A\vdash \qexists x B}{\lrule{$\exists$-R}}
\]
every Opponent move $O$ interpreting an universal quantification $\forall x$ on
the right-hand side of a sequent, or an existential quantification $\exists x$
on the left-hand side of a sequent, is such that $O\leq_\sigma P$ whenever the
variable $x$ is free in the term $t$.
For example, a proof
\[
\inferrule{
\inferrule{
\inferrule{
\inferrule{\null}
{P\vdash Q[t/z]}{\lrule{Ax}}
}
{P\vdash\qexists z Q}{\lrule{$\exists$-R}}
}
{\qexists y P\vdash\qexists z Q}{\lrule{$\exists$-L}}
}
{\qexists x{\qexists y P}\vdash\qexists z Q}{\lrule{$\exists$-L}}
\]
is interpreted respectively by the strategies
\begin{equation}
  \label{eq:ex-intp}
  \sstrid{strat_ex_xy}
  \ \ \ 
  \sstrid{strat_ex_x}
  \ \ \ 
  \sstrid{strat_ex_}
\end{equation}
when the free variables of $t$ are $\{x,y\}$, $\{x\}$ or $\emptyset$.

\begin{remark}
  This interpretation could be generalized to proofs with cuts using the
  composition of the category~$\Games$, and one could show that the
  interpretation is invariant under cut-elimination. However, we do not detail
  this here since it is best expressed using connectives and leave this for
  future works.
\end{remark}


\paragraph{An equational theory of strategies.}
\label{subsection:walking-inno}
We can now introduce the equational theory which will be shown to present the
category~$\Games$.

\begin{definition}
  \label{definition:innocent-strategies}
  The \emph{equational theory of strategies} is the equational theory $\eqth{G}$
  with two atomic types $O$ and $P$ and thirteen generators depicted in
  Figure~\ref{fig:inno-gen} such that
  \begin{itemize}
  \item the Opponent structure
    \begin{equation}
      \label{eq:O-struct}
      (O,\mu^O,\eta^O,\delta^O,\varepsilon^O,\gamma^O)
    \end{equation}
    is a bicommutative qualitative bialgebra,
  \item the object $P$ is left dual to the object $O$ with $\eta^{OP}$ as unit
    and $\varepsilon^{OP}$ as counit,
  \item the Proponent structure
    $(P,\mu^P,\eta^P,\delta^P,\varepsilon^P,\gamma^P)$, as well as the
    morphism~$\gamma^{OP}$, are deduced from the Opponent structure
    \eqref{eq:O-struct} by composition with the duality morphisms~$\eta^{OP}$
    and~$\varepsilon^{OP}$, in the sense that the equations of
    Figure~\ref{fig:PO-adj} hold.
  \end{itemize}
\end{definition}
\begin{figure}[htp!]
  \centering
  \[
  \begin{array}{r@{\qeq}l}
    \strid{mult_P}
    &
    \strid{comult_O_adj}
    \\[4ex]
    \strid{comult_P}
    &
    \strid{mult_O_adj}
    \\[4ex]
    \strid{unit_P}
    &
    \strid{counit_O_adj}
    \\[4ex]
    \strid{counit_P}
    &
    \strid{unit_P_adj}
    \\[4ex]
    \strid{sym_P}
    &
    \strid{sym_O_adj}
    \\[10ex]
    \strid{sym_OP}
    &
    \strid{sym_O_adj_OP}
  \end{array}
  \]
  \caption{Proponent is left dual to Opponent.}
  \label{fig:PO-adj}
\end{figure}
\noindent
We write $\catquot{\mathcal{G}}$ for the monoidal category generated by
$\eqth{G}$. It can be noticed that the generators $\mu^P$, $\eta^P$, $\delta^P$,
$\varepsilon^P$, $\gamma^P$ and~$\gamma^{OP}$ are superfluous in this
presentation (since they can be deduced from the Opponent structure and
duality). However, removing them would seriously complicate the proofs.

\begin{remark}
  If we adopt the point of view of logic, the relations of
  Figure~\ref{fig:PO-adj} (as well as in fact all the relations of our
  presentation) can be understood as rules for cut-elimination. For example,
  suppose for clarity that function symbols include a nullary symbol~$0$, that
  proposition symbols include a nullary symbol~$\top$ and a binary symbol~$=$,
  and that the set~$\axioms$ of axioms contains the reasonable axioms for
  equality, \eg $(\top,x=x)\in\axioms$, etc. In the third equation of
  Figure~\ref{fig:PO-adj}, the left and right members are respectively the
  interpretation of the proofs
  \[
  \inferrule{
    \inferrule{\null}
    {\top\vdash 0=0}{\lrule{Ax}}
  }
  {\top\vdash\qexists x x=0}{\lruler\exists}
  \]
  and
  \[
  \inferrule{
    \inferrule{
      \inferrule{
        \inferrule{\null}
        {\top\vdash y=y}{\lrule{Ax}}
      }
      {\top\vdash\qexists z y=z}{\lruler\exists}
    }
    {\top\vdash\qforall y{\qexists z y=z}}{\lruler\forall}
    \\
    \inferrule{
      \inferrule{
        \inferrule{
          \inferrule{\null}
          {0=z\vdash z=0}{\lrule{Ax}}
        }
        {0=z\vdash\qexists x x=0}{\lruler\exists}
      }
      {\qexists z 0=z\vdash\qexists x x=0}{\lrulel\exists}
    }
    {\qforall y{\qexists z y=z}\vdash\qexists x x=0}{\lrulel\exists}
  }
  {\top\vdash\qexists x x=0}{\lrule{Cut}}
  \]
  and the second proof reduces to the first one by cut-elimination.
\end{remark}

\begin{lemma}
  With the notations of~\ref{definition:innocent-strategies}, we have:
  \begin{itemize}
  \item $(P,\mu^P,\eta^P,\delta^P,\varepsilon^P,\gamma^P)$ is a qualitative
    bicommutative bialgebra,
  \item the Yang-Baxter equalities
    \[
    \strid{yang_baxter_xyz_r}
    \qeq
    \strid{yang_baxter_xyz_l}
    \]
    hold whenever $(X,Y,Z)$ is either $(O,O,O)$, $(P,O,O)$, $(P,P,O)$ or
    $(P,P,P)$,
  \item the equalities
    \[
      \strid{mult_sym_rnat_P_l}
      =
      \strid{mult_sym_rnat_P_r}
    \]
    and
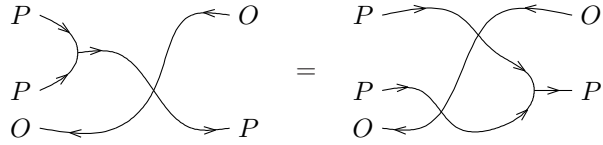
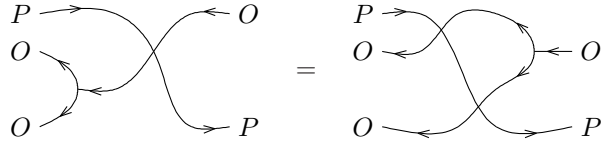
    \[
    \strid{mult_sym_lnat_O_l}
    =
    \strid{mult_sym_lnat_O_r}
    \]
    hold (and dually for comultiplications),
  \item the equalities
    \[
    \strid{eta_sym_rnat_P_l}
    =
    \strid{eta_sym_rnat_P_r}
    \]
    and
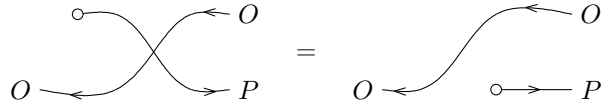
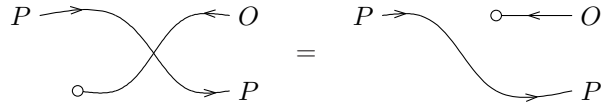
    \[
    \strid{eta_sym_lnat_O_l}
    =
    \strid{eta_sym_lnat_O_r}
    \]
    hold (and dually for counits),
  \item the equalities
    \[
    \strid{adj_counit_O_r}
    =
    \strid{adj_counit_O_l}
    \]
    and
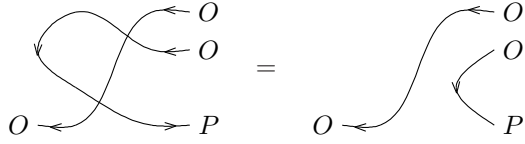
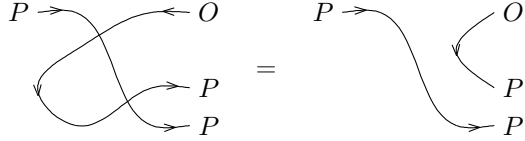
    \[
    \strid{adj_counit_P_r}
    =
    \strid{adj_counit_P_l}
    \]
    hold (and dually for the counit of duality).
  \end{itemize}
\end{lemma}

We can now proceed as in Section~\ref{section:presentation-rel} to show that the
theory~$\eqth{G}$ introduced in Definition~\ref{definition:innocent-strategies}
presents the category~$\Games$. First, in the category $\Games$ with the
monoidal structure induced by $\before{}$, the objects $O$ and $P$ can be
canonically equipped with thirteen morphisms as shown in
Figure~\ref{fig:inno-gen} in order to form a model of the theory~$\eqth{G}$.

Conversely, we need to introduce a notion of canonical form for the morphisms
of~$\mathcal{G}$. Stairs are defined similarly as before, but are now
constructed from the three kinds of polarized crossings $\gamma^O$, $\gamma^P$
and $\gamma^{OP}$ instead of simply~$\gamma$ in~\eqref{eq:ctx-S}: a \emph{stair}
is either $\id_O$ or~$\id_P$
or
\[
\strid{gsym_sym_O}
\tor
\strid{gsym_sym_P}
\tor
\strid{gsym_sym_OP}
\]
The notion of \emph{precanonical form}~$\phi$ is now defined inductively as
shown in Figure~\ref{fig:precan-strat},
\begin{figure}[htp!]
  \centering
    \vbox{
      $\phi$ is either empty or
      \[
        \begin{array}{c}
          A_i\phi'\quad=\sstrid{nf_adj}
          \qtor
          B_i\phi'\quad=\sstrid{nf_coadj}
          \\[37ex]
          \qtor
          H^X\phi'\quad=\sstrid{nf_eta}
          \qtor
          E^X\phi'\quad=\sstrid{nf_eps}
          \\[20ex]
          \qtor
          W_i\phi'\quad=\sstrid{nf_mu}
        \end{array}
      \]
    }
  \caption{Precanonical forms for strategies.}
  \label{fig:precan-strat}
\end{figure}
where the object~$X$ is either~$O$ or~$P$ and $\phi'$ is a precanonical
form. These cases correspond respectively to the productions of the following
grammar
\[
\phi
\qgramdef
Z
\sgramor
A_i\phi
\sgramor
B_i\phi
\sgramor
W_i\phi
\sgramor
E^X\phi
\sgramor
H^X\phi
\]
By induction on the size of morphisms, it can be shown that every morphism
of~$\mathcal{G}$ is equivalent to a precanonical form and a notion of canonical
form can be defined by adapting the rewriting system~\eqref{eq:mrel-cf-rs} into
a rewriting system for precanonical forms, by adding the rules
\[
\begin{array}{r@{\qquad\Longrightarrow\qquad}l@{\qquad}l}
  H^XW_i&W_{i+1}H^X\\
  H^XE^Y&E^YH^X\\
  W_iW_j&W_jW_i&\text{when $i<j$}\\
  W_iW_i&W_i\\
  H^XA_i&A_iH^X\\
  A_iW_j&W_jA_i\\
  A_iA_j&A_jA_i&\text{when $i<j$}\\
  A_iA_i&A_i\\
  E^XB_i&E^X\\
  B_iW_j&W_jB_i\\
  B_iB_j&B_jB_i&\text{when $i<j$}\\
  B_iB_i&B_i\\
  B_iA_j&A_jB_i\\
\end{array}
\]
to the rewriting system containing the rules~\eqref{eq:mrel-cf-rs}
and~\eqref{eq:rel-cf-rs}. It is simple to extend the proof of
Lemma~\ref{lemma:mrel-can} in order to show that this rewriting system is
normalizing. The general form for canonical forms is
\begin{equation}
  \label{eq:games-explicit-cf}
  \begin{array}{r@{}l}
    W_{i^n_{k_n}}\cdots W_{i^n_1}A_{j^n_{l_n}}\cdots A_{j^n_1}E\cdots&\cdots W_{i^1_{k_1}}\cdots W_{i^1_1}A_{j^1_{l_1}}\cdots A_{j
^1_1}E\\
    &\cdots B_{h^p_{m_p}}\cdots B_{h^p_1}H\cdots B_{h^1_{m_1}}\cdots B_{h^1_1}HZ
  \end{array}
\end{equation}
with
\begin{itemize}
\item $i^p_{k_p}>\ldots>i^p_1$ for every integer $r$ such that $1\leq r\leq k_n$,
\item $j^p_{l_p}>\ldots>j^p_1$ for every integer $r$ such that $1\leq r\leq l_n$,
\item $h^p_{l_p}>\ldots>h^p_1$ for every integer $r$ such that $1\leq r\leq m_n$.
\end{itemize}


\begin{lemma}
  Every strategy $\sigma:A\to B$ is the interpretation of an unique canonical
  form.
\end{lemma}
\begin{proof}
  We show that every strategy~$\sigma:A\to B$ is the interpretation of a
  precanonical form~$\phi:A\to B$ by induction on the
  triple~$(\size{A},\size{\sigma},\size{B})$, ordered lexicographically.
  \begin{enumerate}
  \item If~$A=B=I$ then~$\sigma$ is the interpretation of the precanonical
    form~$Z$.
  \item If~$A=I$ and~$B=X\otimes B'$, where~$X$ is either~$P$ or~$O$ then we
    distinguish two cases.
    \begin{itemize}
    \item If no move depends on~$X$ in the strategy, this strategy is the image
      of a precanonical form~$H_X\phi'$, where $\phi'$ is a precanonical form,
      obtained by induction hypothesis whose interpretation is the strategy
      \hbox{$\sigma':I\to B'$} obtained by restricting~$\sigma$ to the
      codomain~$B$ (the size of~$\sigma'$ is~$\size{\sigma'}=\size\sigma$).
    \item Otherwise, we write~$i$ for the index in~$B$ of the move of minimal
      index which depends on~$X$ in the strategy. The strategy is the image of a
      precanonical form~$B_i\phi'$, where~$\phi'$ is precanonical form, obtained
      by induction hypothesis, whose interpretation is the
      strategy~$\sigma':I\to B$ obtained from~$\sigma$ by removing the
      dependency of the $i$-th move of~$B$ on the first move of~$B$ (its size is
      such that~$\size{\sigma'}<\size\sigma$).
    \end{itemize}
  \item If~$A=X\otimes A'$, where~$X$ is either~$P$ or~$O$, then we distinguish
    three cases.
    \begin{itemize}
    \item If no move depends on~$X$ in the strategy, this strategy is the image
      of a precanonical form~$E^X\phi'$, where~$\phi'$ is a precanonical form,
      obtained by induction hypothesis, whose interpretation is the strategy
      \hbox{$\sigma':A'\to B$} obtained by restricting~$\sigma$ to the
      domain~$A'$.
    \item If there exists a move of~$X$ which depends on~$X$, we write~$i$ for
      the index in~$A$ of such a move of minimal index. The strategy is the
      interpretation of a precanonical form~$A_i\phi'$, where~$\phi'$ is a
      precanonical form, obtained by induction hypothesis, whose interpretation
      is the strategy~$\sigma':A\to B$ obtained from~$\sigma$ by removing the
      dependency of the $i$-th move of~$A$ on the first move of~$A$ (its size is
      such that \hbox{$\size{\sigma'}<\size\sigma$}).
    \item Otherwise, there exists a move in~$B$ which depends on the
      move~$X$. We write~$i$ of the index in~$B$ of such a move of minimal
      index. The strategy is the interpretation of a precanonical
      form~$W_i\phi'$, where~$\phi'$ is a precanonical form, obtained by
      induction hypothesis, whose interpretation is the strategy~$\sigma':A\to
      B$, obtained from~$\sigma$ by removing the dependency of the $i$-th move
      of~$B$ on the first move of~$A$ (its size is such
      that~$\size{\sigma'}<\size\sigma$).
    \end{itemize}
  \end{enumerate}
  Knowing the general form~\eqref{eq:games-explicit-cf} of canonical forms, it
  is easy to show that the precanonical forms thus constructed are actually
  canonical and that canonical forms~$\phi:A\to B$ are in bijection with
  strategies~$\sigma:A\to B$, as in the proof of Theorem~\ref{thm:fmr-pres}.
\end{proof}

We therefore deduce the main theorem of this article:

\begin{theorem}
  \label{thm:pres-games}
  The monoidal category $\Games{}$ (with the~$\before{}$ tensor product) is
  presented by the equational theory $\eqth{G}$.
\end{theorem}

\noindent
As a direct consequence of this Theorem, we deduce the two following properties
which show the technical benefits of our construction.

\begin{theorem}
  \label{thm:composition}
  The composite of two strategies, in the sense of
  Definition~\ref{def:strategy}, is itself a strategy (in particular, the
  acyclicity property is preserved by composition).
\end{theorem}
\begin{proof}
  Two strategies~$\sigma:A\to B$ and~$\tau:B\to C$ can be seen as
  morphisms~$\tilde{\sigma}$ and~$\tilde{\tau}$ the
  category~$\catquot{\mathcal{G}}$ and the image of their composite
  is~$\widetilde{\tau\circ\sigma}=\tilde{\tau}\circ\tilde{\sigma}$, which
  corresponds to the image of an unique acyclic strategy.
\end{proof}

\begin{theorem}
  \label{thm:definability}
  The strategies of~$\Games$ are definable (when the set~$\axioms$ of axioms is
  reasonably large enough): it is enough to check that generators are definable
  -- for example, the first case of~\eqref{eq:ex-intp} shows that~$\mu^P$ is
  definable.
\end{theorem}
\begin{proof}
  Suppose that there is a countable number of variable symbols. Suppose moreover
  that there exists a unary propositional symbol~$I$, which enables us to see
  every term~$t$ as a proposition~$I(t)$, which we will simply write~$t$ by
  abuse of notation. We also suppose that the set of propositions contains two
  nullary propositions~$\top$ and~$\bot$ and is closed under formal conjunctions
  and disjunctions: if we have that~$P(x_1,\ldots,x_n)$ and~$Q(y_1,\ldots,y_m)$
  are propositions then \hbox{$P(x_1,\ldots,x_n)\land Q(y_1,\ldots,y_m)$} and
  $P(x_1,\ldots,x_n)\lor Q(y_1,\ldots,y_m)$ are also propositions. We then
  define a set~$\axioms$ of axioms as the smallest set of pairs of propositions
  which is reflexive, transitive and such that:
  \begin{itemize}
  \item for every proposition $P$,
    \begin{itemize}
    \item $(P,\top)\in\axioms$,
    \item $(\bot,P)\in\axioms$,
    \end{itemize}
  \item for every propositions $P$, $P_1$ and $P_2$,
    \begin{itemize}
    \item if $(P,P_1)\in\axioms$ and $(P,P_2)\in\axioms$ then $(P,P_1\land
      P_2)\in\axioms$,
    \item if $(P,P_1)\in\axioms$ or $(P,P_1)\in\axioms$ then $(P,P_1\lor
      P_2)\in\axioms$,
    \item if $(P_1,P)\in\axioms$ or $(P_2,P)\in\axioms$ then $(P_1\land
      P_2,P)\in\axioms$,
    \item if $(P_1,P)\in\axioms$ and $(P_2,P)\in\axioms$ then $(P_1\lor
      P_2,P)\in\axioms$.
    \end{itemize}
  \end{itemize}
  (for concision, we did not mention the arguments of propositions). By
  Theorem~\ref{thm:pres-games}, every strategy can be expressed as a tensor and
  composite of the generating strategies pictured in
  Figure~\ref{fig:inno-gen}. It is therefore enough to show that those
  strategies are definable.
  \begin{itemize}
  \item the strategies~$\mu^P$ and~$\eta^P$ are the respective interpretations
    of the proofs
    \[
    \inferrule{
      \inferrule{
        \inferrule{
          \inferrule{
            \null
          }
          {x\land y\vdash x\land y}{\lrule{Ax}}
        }
        {x\land y\vdash\qexists z z}{\lruler\exists}
      }
      {\qexists y {x\land y}\vdash\qexists z z}{\lrulel\exists}
    }
    {\qexists x{\qexists y {x\land y}}\vdash\qexists z z}{\lrulel\exists}
    \qtand
    \inferrule{
      \inferrule{
        \null
      }
      {\top\vdash\top}{\lrule{Ax}}
    }
    {\top\vdash\qexists x x}{\lruler\exists}
    \]
  \item the strategies~$\delta^P$ and~$\varepsilon^P$ are the respective
    interpretations of the proofs
    \[
    \inferrule{
      \inferrule{
        \inferrule{
          \inferrule{
            \null
          }
          {x\vdash x\land x}{\lrule{Ax}}
        }
        {x\vdash\qexists z{x\land z}}{\lruler\exists}
      }
      {x\vdash\qexists y{\qexists z{y\land z}}}{\lruler\exists}
    }
    {\qexists x x\vdash\qexists y{\qexists z{y\land z}}}{\lrulel\exists}
    \qtand
    \inferrule{
      \inferrule{
        \null
      }
      {x\vdash\top}{\lrule{Ax}}
    }
    {\qexists x x\vdash\top}{\lrulel\exists}
    \]
  \item the strategies~$\eta^{OP}$ and~$\varepsilon^{OP}$ are the respective
    interpretations of the proofs
    \[
    \inferrule{
      \inferrule{
        \inferrule{
          \null
        }
        {\top\vdash x\lor(x\lor\top)}{\lrule{Ax}}
      }
      {\top\vdash\qexists y{x\lor y}}{\lruler\exists}
    }
    {\top\vdash\qforall x{\qexists y{x\lor y}}}{\lruler\forall}
    \qtand
    \inferrule{
      \inferrule{
        \inferrule{
          \null
        }
        {x\land(x\land\bot)\vdash\bot}{\lrule{Ax}}
      }
      {\qforall y{x\land y}\vdash\bot}{\lrulel\forall}
    }
    {\qexists x{\qforall y{x\land y}}\vdash\bot}{\lrulel\exists}
    \]
  \item the strategies~$\gamma^P$ and~$\gamma^{OP}$ are the respective
    interpretations of the proofs
    \[
    \inferrule{
      \inferrule{
        \inferrule{
          \inferrule{
            \inferrule{
              \null
            }
            {x\land y\vdash x\land y}{\lrule{Ax}}
          }
          {x\land y\vdash\qexists t{t\land y}}{\lruler\exists}
        }
        {x\land y\vdash\qexists z{\qexists t{t\land z}}}{\lruler\exists}
      }
      {\qexists y{x\land y}\vdash\qexists z{\qexists t{t\land z}}}{\lrulel\exists}
    }
    {\qexists x{\qexists y{x\land y}}\vdash\qexists z{\qexists t{t\land z}}}{\lrulel\exists}
    \qtand
    \inferrule{
      \inferrule{
        \inferrule{
          \inferrule{
            \inferrule{
              \null
            }
            {x\land z\vdash x\land z}{\lrule{Ax}}
          }
          {x\land z\vdash\qexists t{t\land z}}{\lruler\exists}
        }
        {\qforall y{x\land y}\vdash\qexists t{t\land z}}{\lrulel\forall}
      }
      {\qexists x{\qforall y{x\land y}}\vdash\qexists t{t\land z}}{\lrulel\exists}
    }
    {\qexists x{\qforall y{x\land y}}\vdash\qforall z{\qexists t{t\land z}}}{\lruler\forall}
    \]
  \item etc.\qedhere
  \end{itemize}
\end{proof}

A given strategy is not necessarily the interpretation of a unique proof. In
particular, as explained in the introduction, two proofs which only differ by
the order of introduction of some successive connectives are identified in the
semantics.

In the preceding proof, we could of course have taken the set of all pairs of
propositions as set~$\axioms$ of axioms. The set that we have used shows however
that our definability result can be obtained with a reasonable set of axioms: it
is in particular \emph{coherent}, which means that there exists a sequent which
cannot be proved (the sequent $\top\vdash\bot$ for example), which would not
have been the case with the trivial set of axioms.

\section{Conclusion}
We have constructed a game semantics for first-order propositional logic and
given a presentation of the category~$\Games$ of games and definable
strategies. This has revealed the essential structure of causality induced by
quantifiers as well as provided technical tools to show definability and
composition of strategies.

We consider this work much more as a starting point to bridge semantics and
algebra than as a final result. The methodology presented here seems to be very
general and many tracks remain to be explored.

First, we would like to extend the presentation to a game semantics for richer
logic systems, containing connectives (such as conjunction or
disjunction). Whilst we do not expect essential technical complications, this
case is much more difficult to grasp and manipulate, since a presentation of
such a semantics would have generators up to dimension~3: games would be modeled
as trees of connectives and strategies as ``surface diagrams'' between these
trees. It would be particularly interesting to do this for the multiplicative
fragment of linear logic (MLL) with first-order quantifiers since it would
provide us with a local reformulation of the Danos-Regnier criterion for MLL
extended with the MIX rule (this is hinted in Remark~\ref{rem:ll-acyclicity}).

Some of the proofs (such as the proof of Lemma~\ref{lemma:mrel-precan}) are very
repetitive, which we think is a good point: we believe that they could be
mechanically checked or automated. It turns out that it is quite difficult to
find a good representation of morphisms in monoidal categories, which is
suitable for a computer to manipulate them without having to handle complex
congruences such as the exchange law. We have proposed such a representation as
well as an unification algorithm for monoidal rewriting
systems~\cite{mimram:rta10}, but many properties and generalizations of these
techniques remain to be investigated in order to have really useful
tools. Formulated in categorical terms this amounts to generalize term rewriting
techniques from Lawvere theories (which are categories with products, thus
monoidal categories, thus 2-categories with one object) to the general setting
of 2-categories. In particular, it would also be interesting to know whether it
is possible to orient the equalities in the presentations in order to obtain
strongly normalizing rewriting systems for the algebraic structures described in
the paper. Such rewriting systems are given in~\cite{lafont:boolean-circuits},
for monoids and commutative monoids, etc., but for example finding a strongly
normalizing rewriting system presenting the theory of bialgebras is a difficult
problem~\cite{mimram:phd}, not to mention a strongly normalizing presentation of
our category of games. Such a presentation would have a very high number of
critical pairs which makes us see the development of automated tools to compute
them a necessary preliminary step.

Finally, there is a striking analogy between the string diagrams we have used
and wires in electronic circuits. This is actually one of the starting point of
the current work of Ghica (as well as game semantics), who is currently
elaborating a compiler from a high-level language into integrated
circuits~\cite{ghica:geometry-synthesis}. The categorical string-diagrammatic
axioms reveal to be crucial in this setting in order to establish designing
principles for the circuits. Following this point of view, we believe that a
deep understanding of the algebraic structure of categories of semantics of
programming languages will prove very useful in order to design and optimize
circuits implementing programs in these languages.

\paragraph{Acknowledgments.}
I would like to thank Martin Hyland and Paul-André Melliès, as well as John
Baez, Albert Burroni, Jonas Frey, Yves Guiraud, Yves Lafont, François Métayer
and Luke Ong, for the lively discussion we had, during which I learned so much;
I also thank the anonymous referee for valuable suggestions.

\bibliographystyle{alpha}
\bibliography{these}

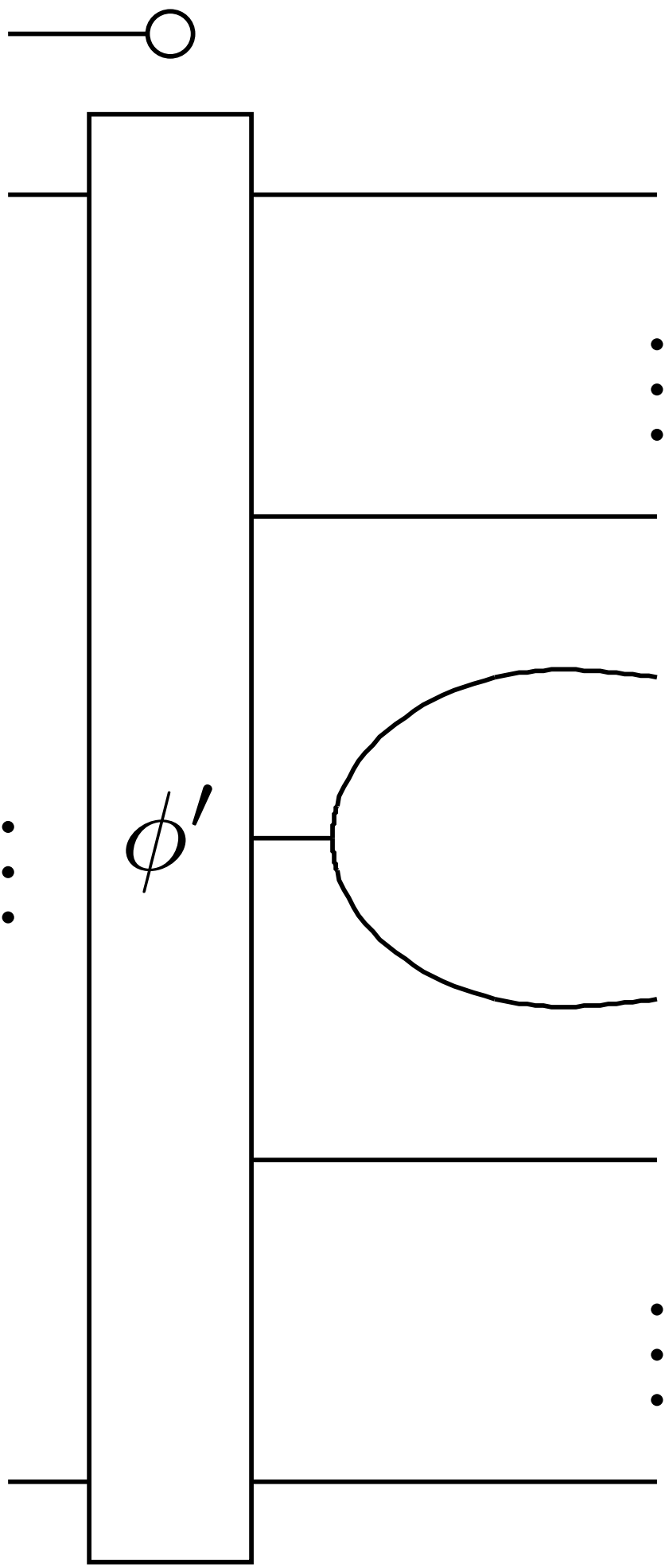
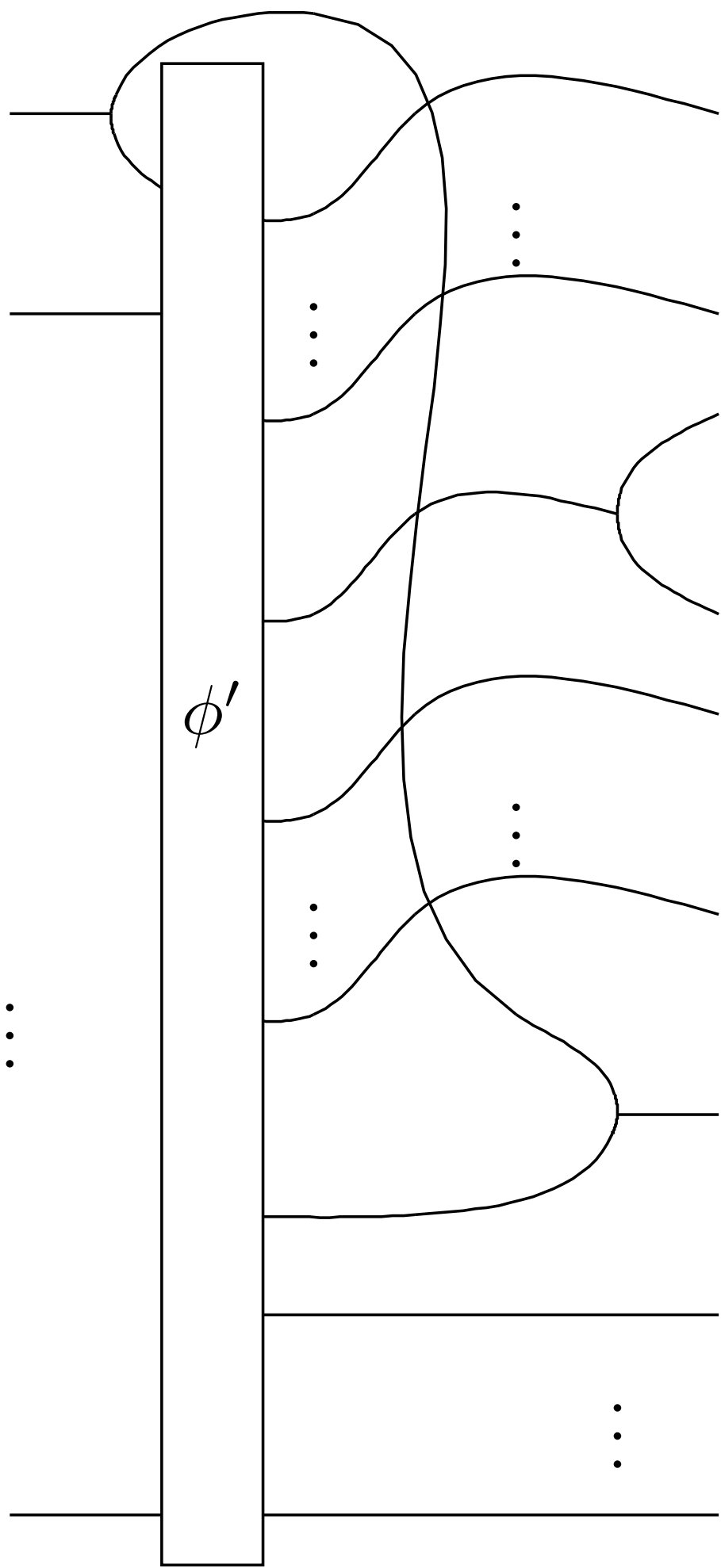
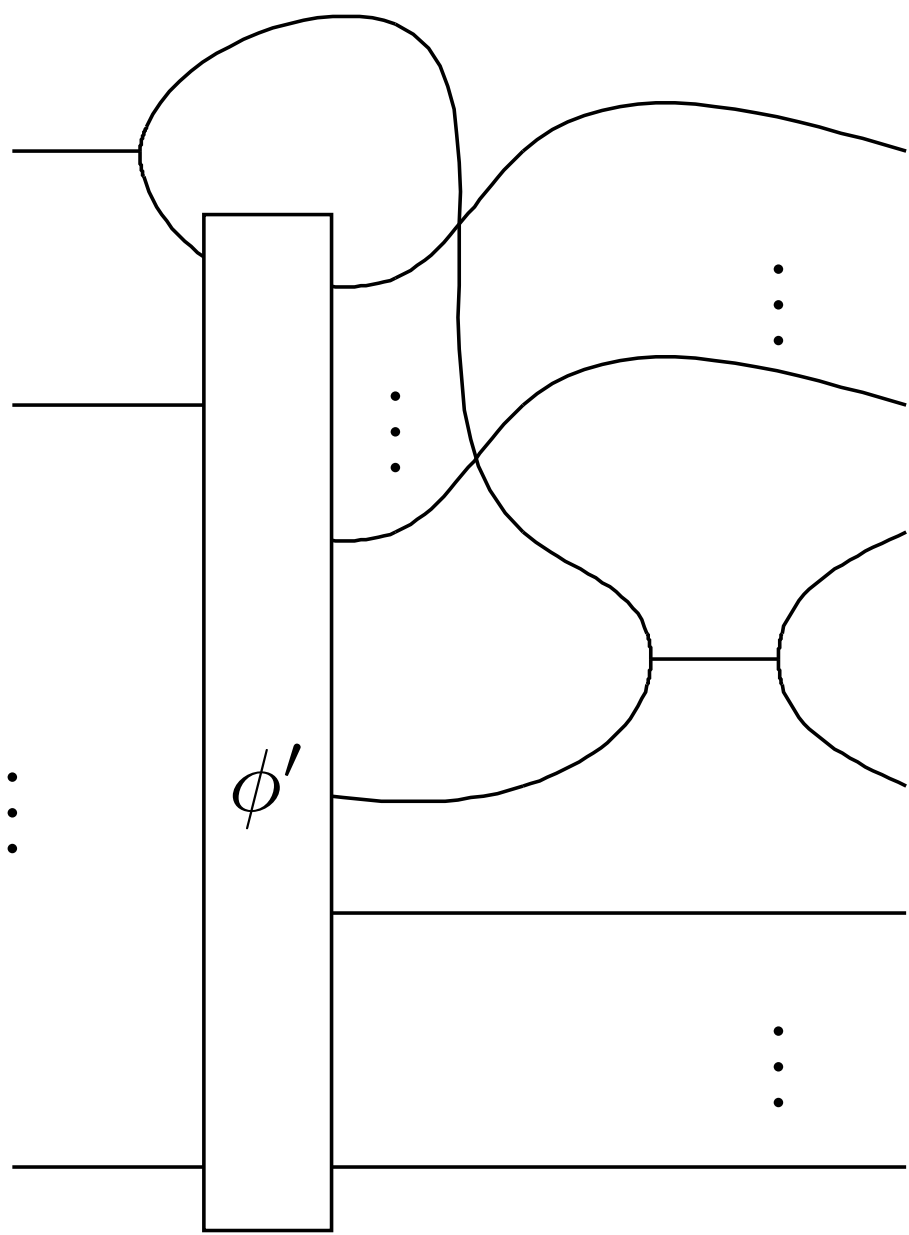
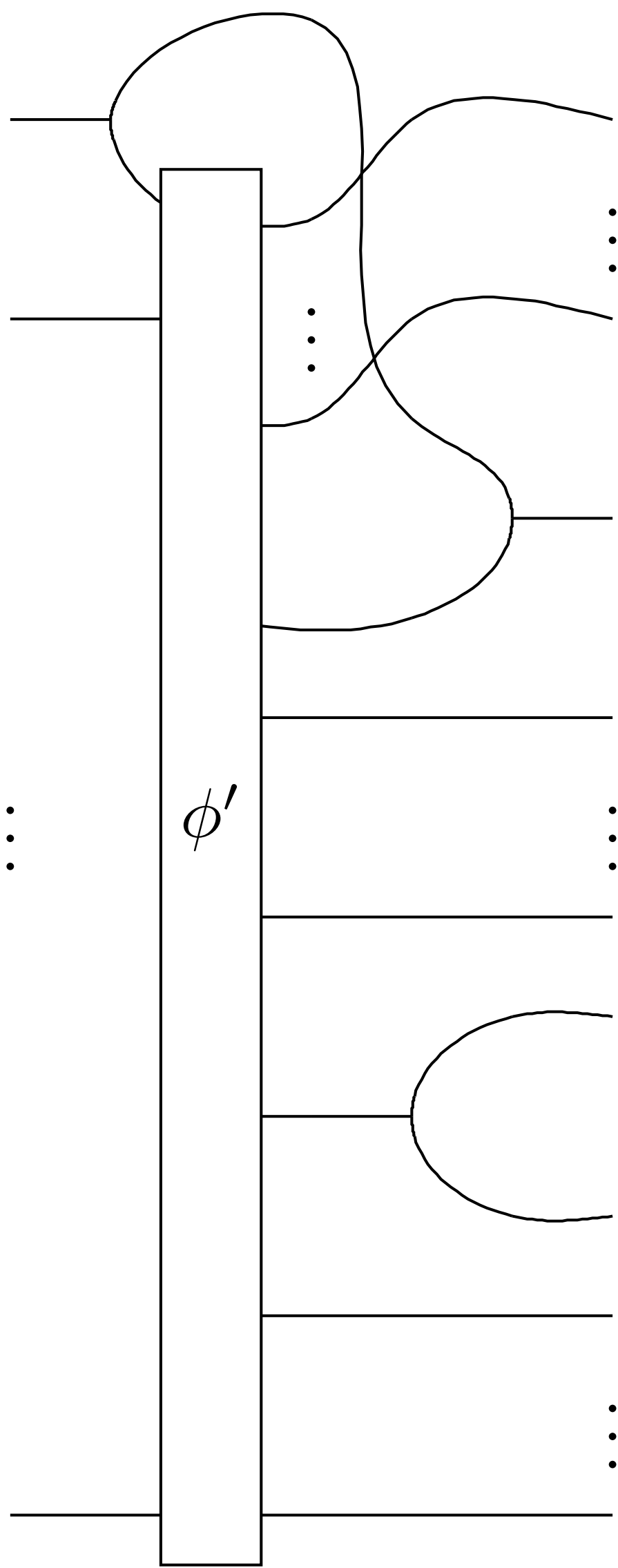
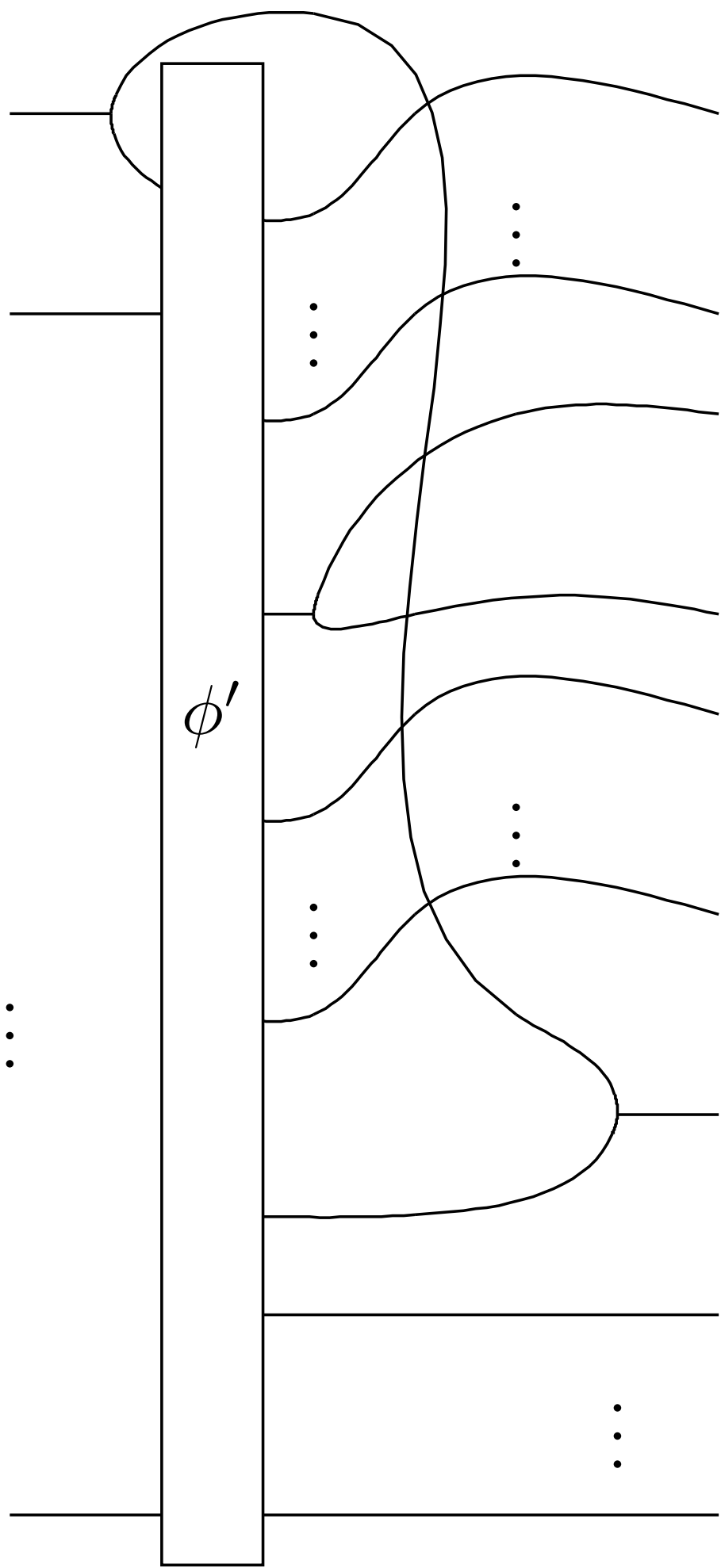
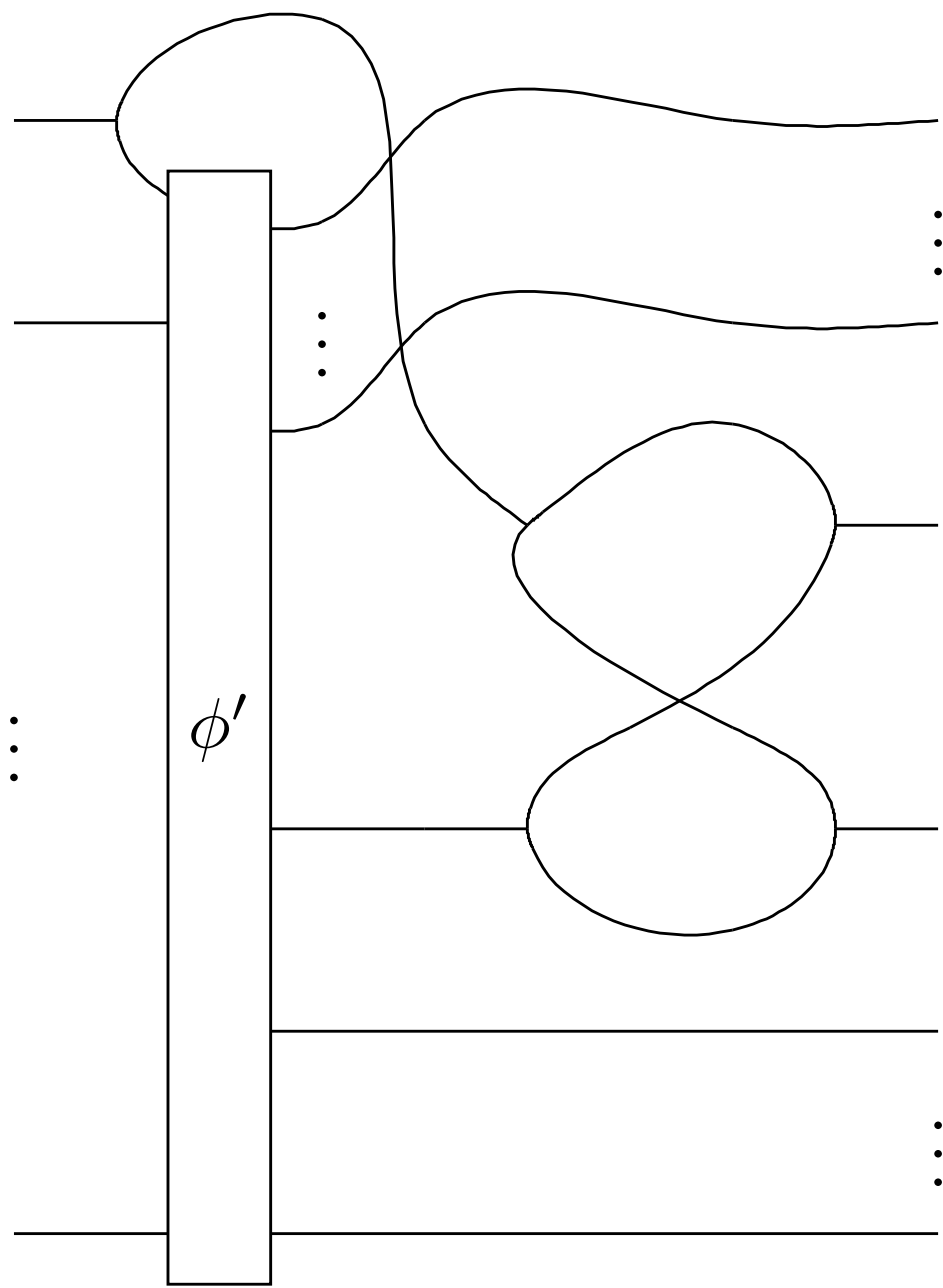
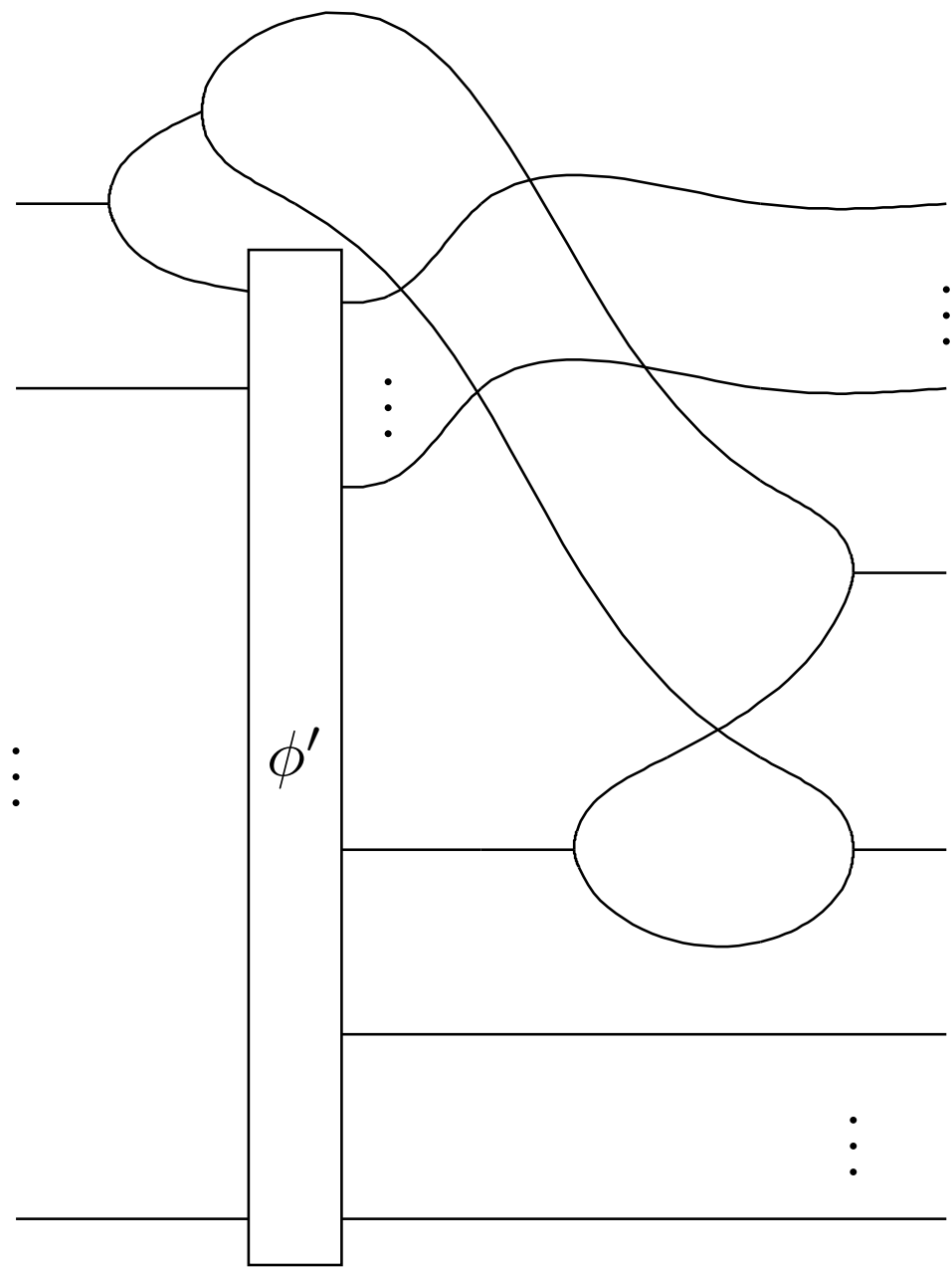
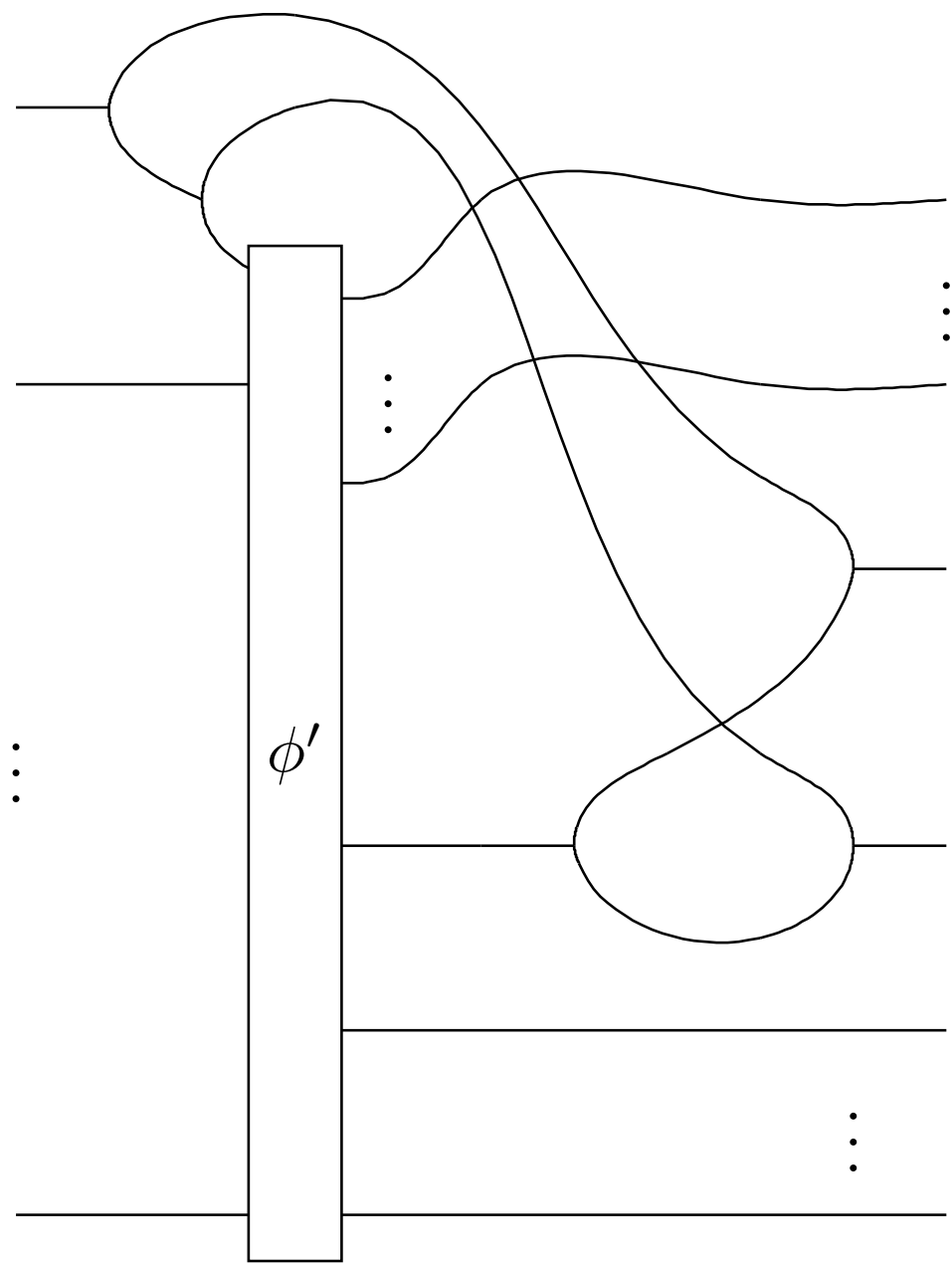
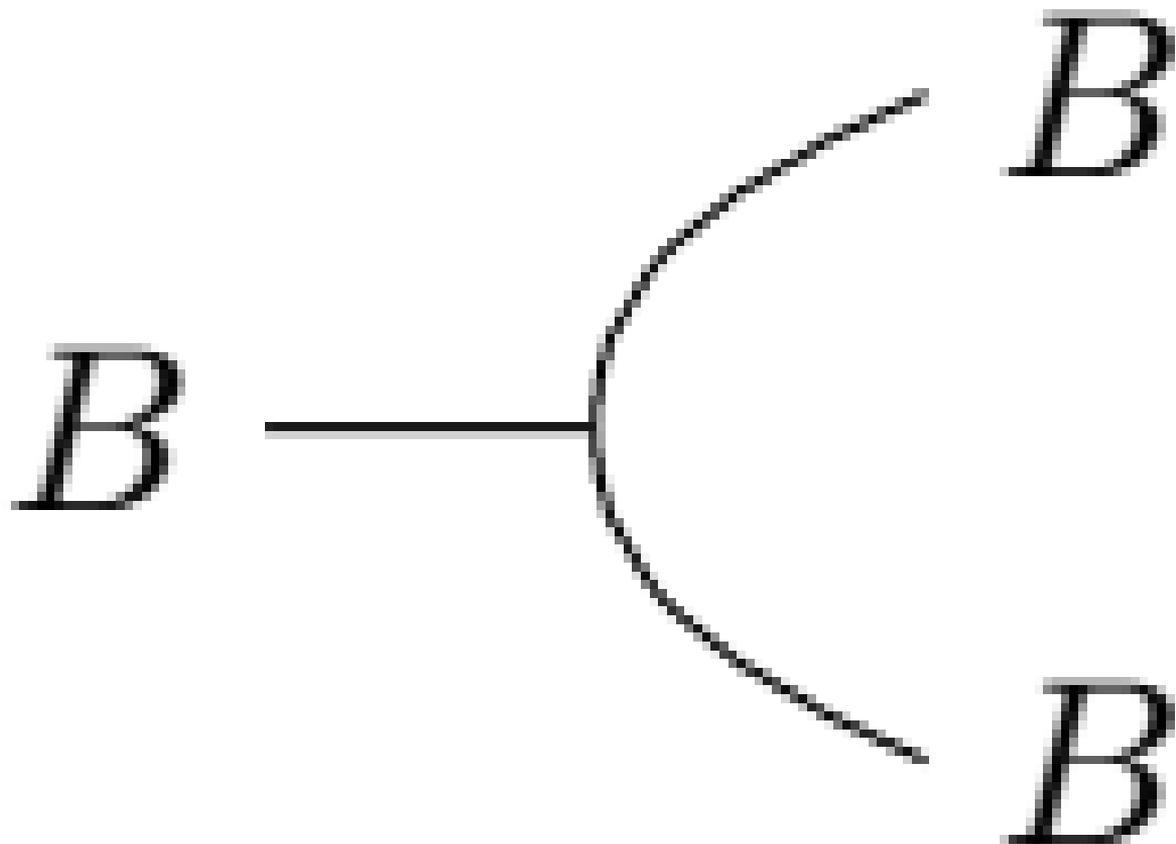
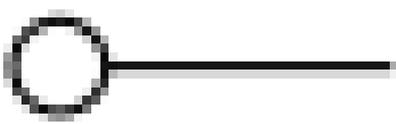
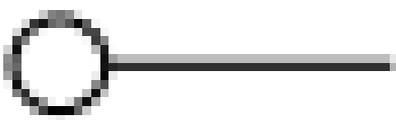
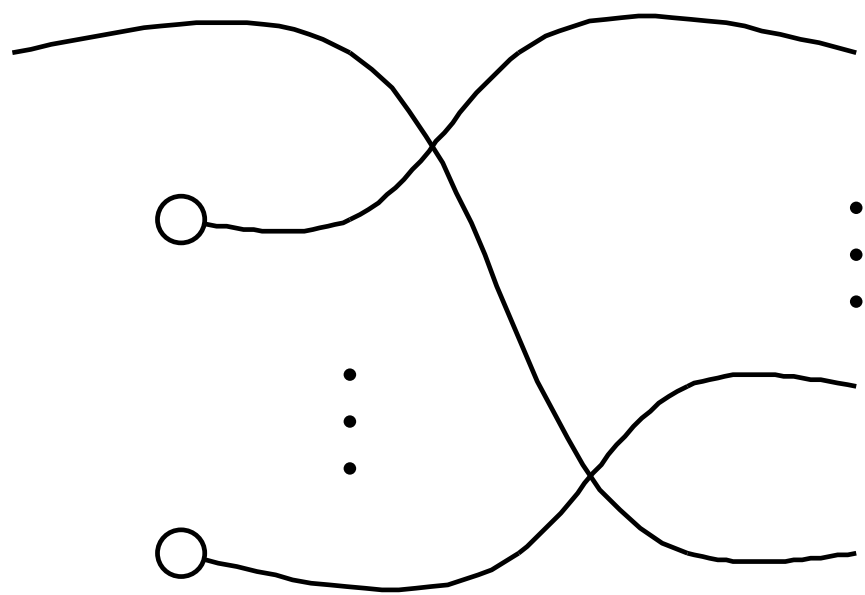
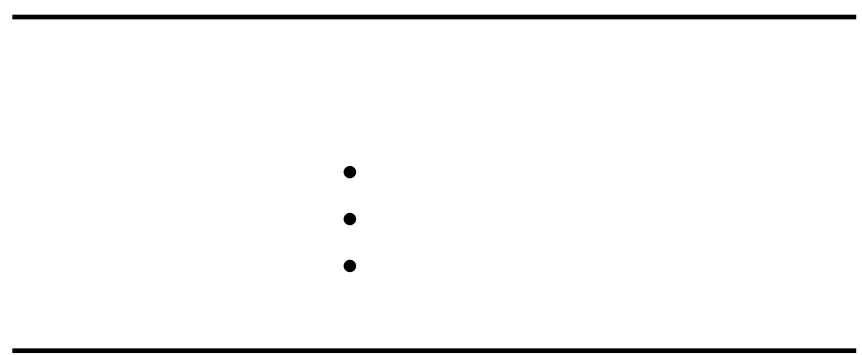
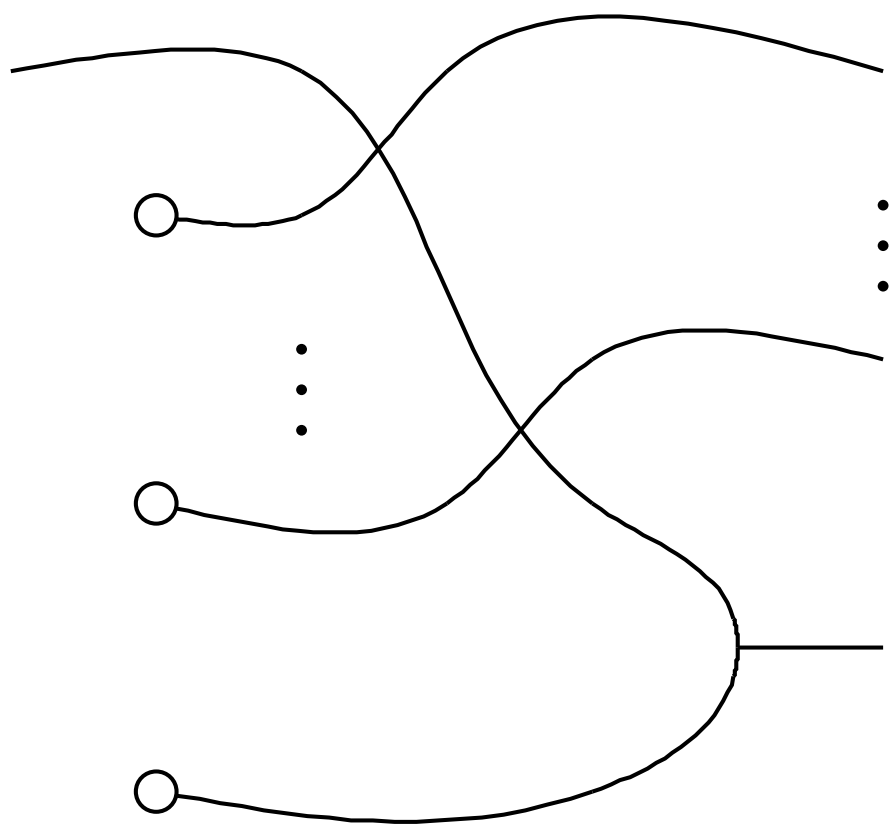
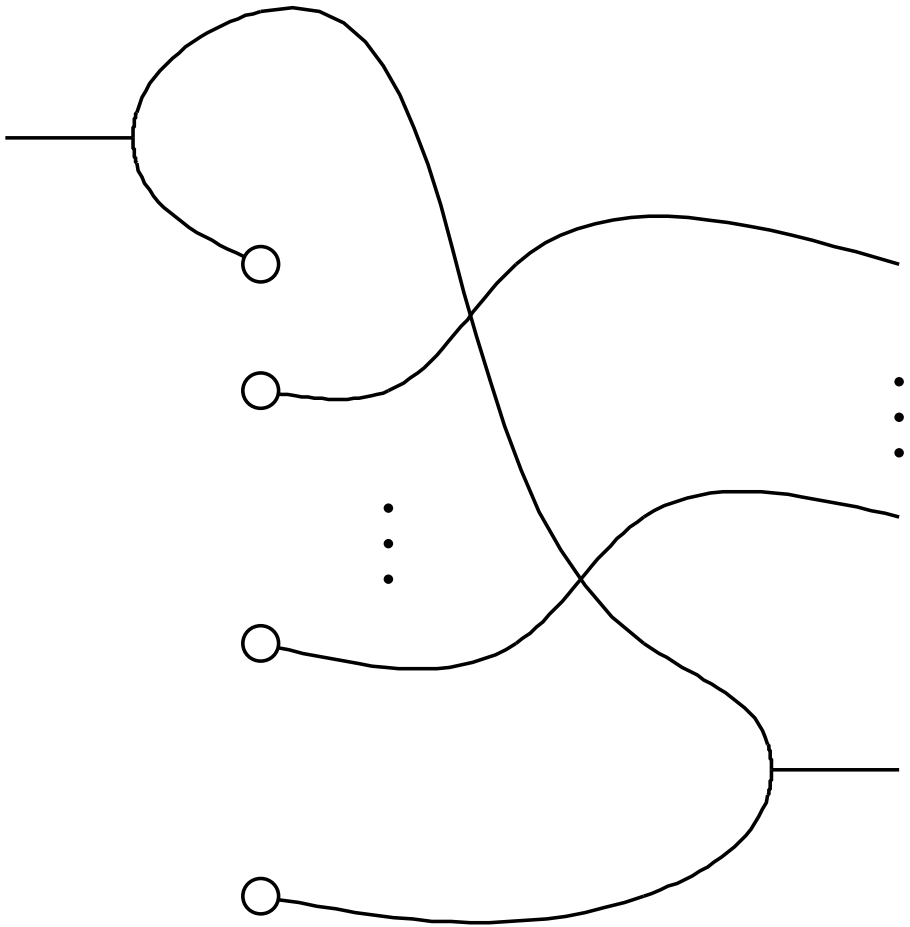
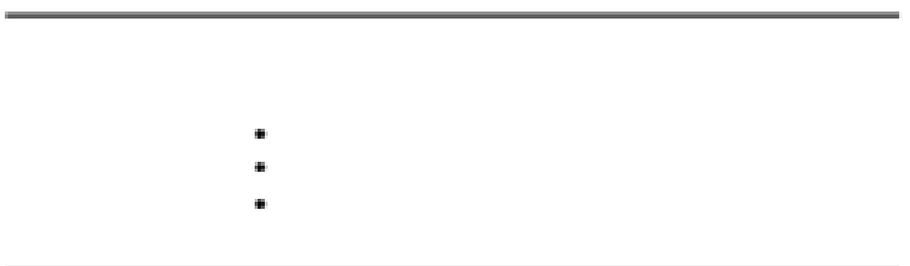
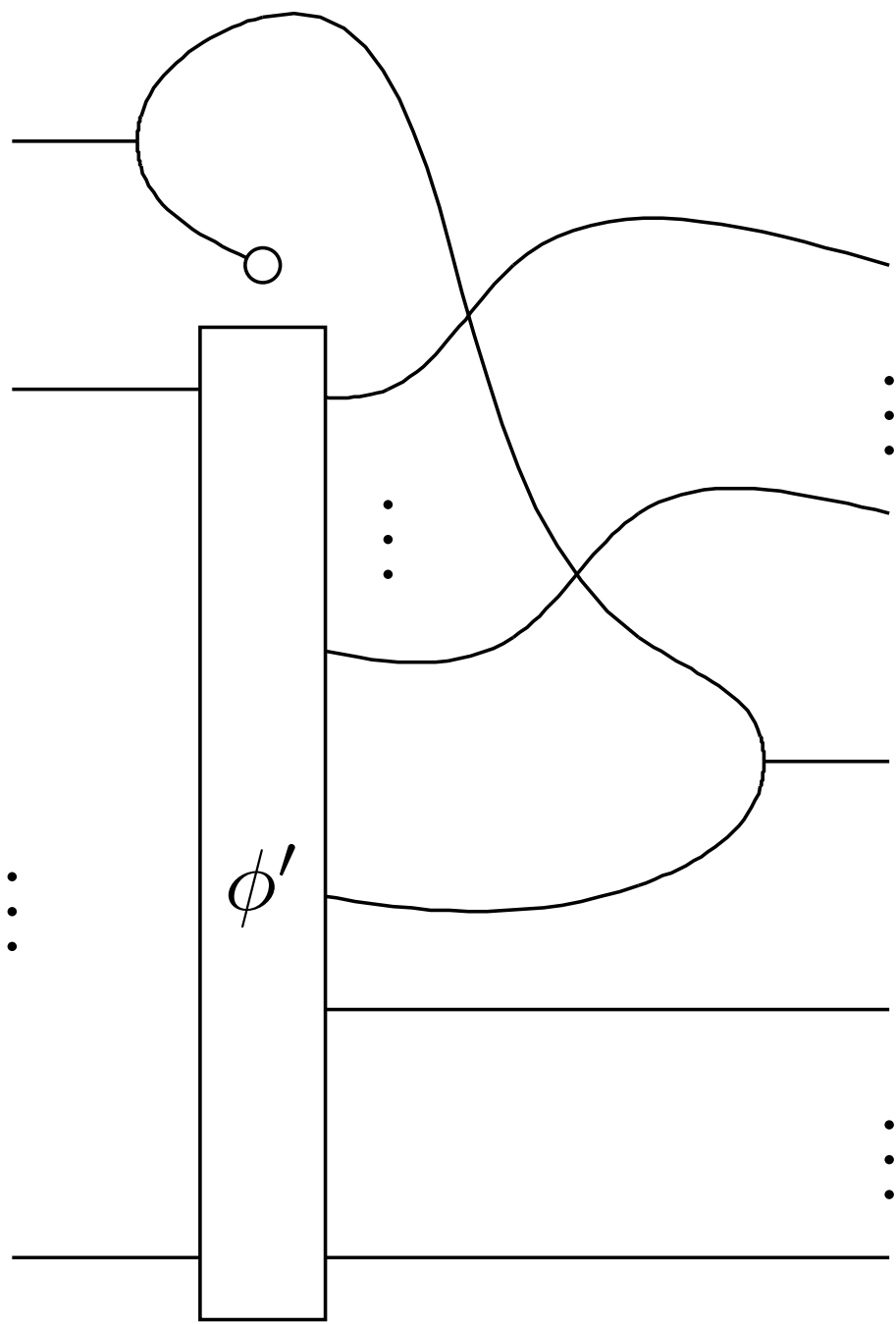
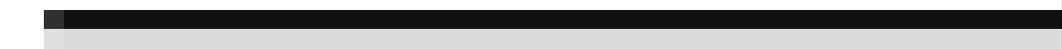
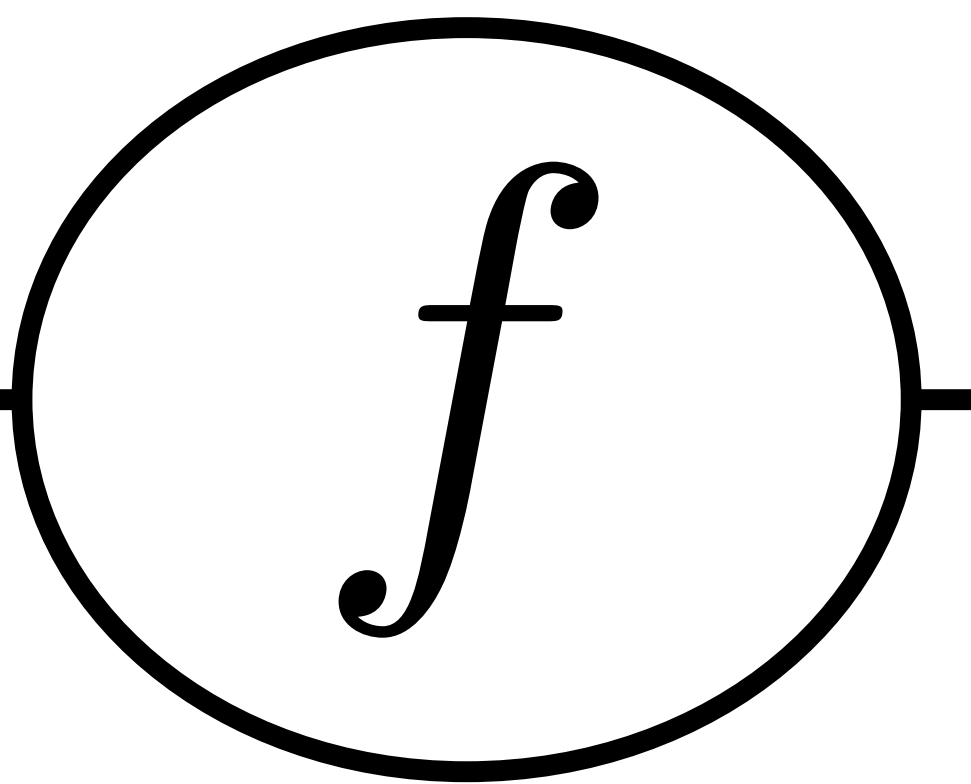
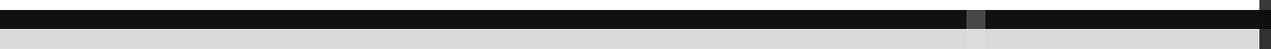
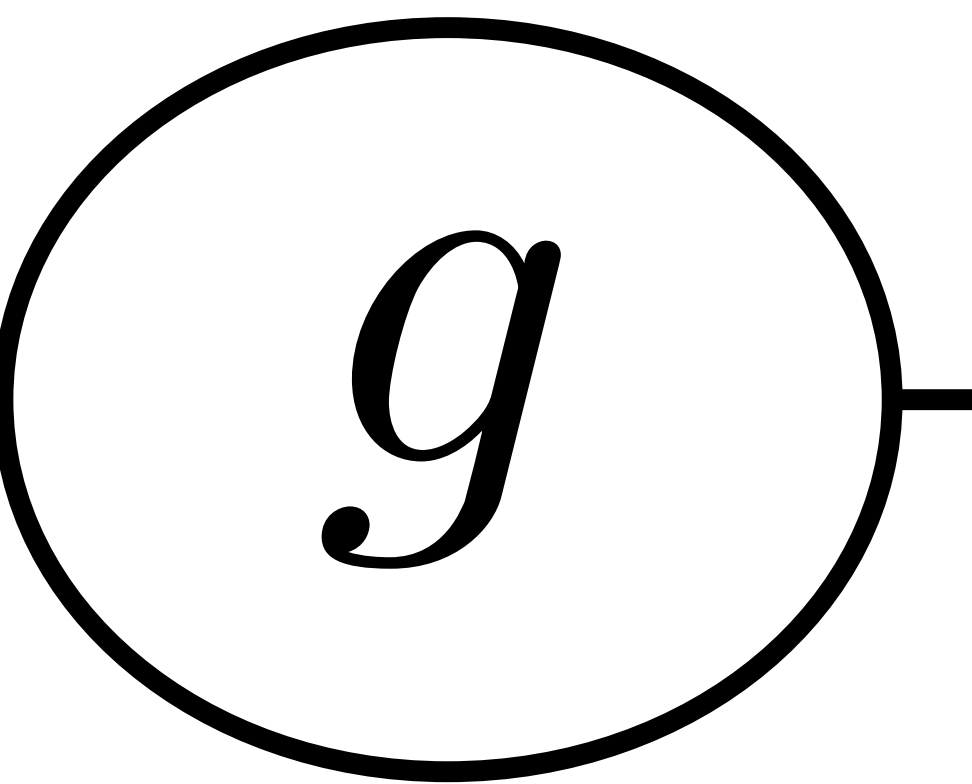
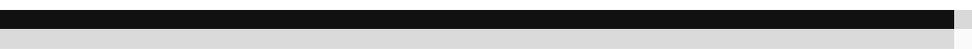
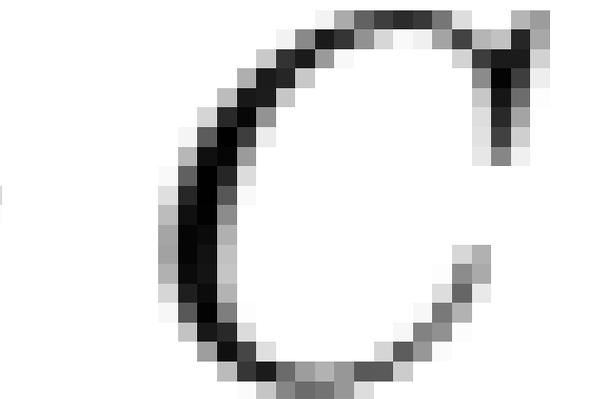
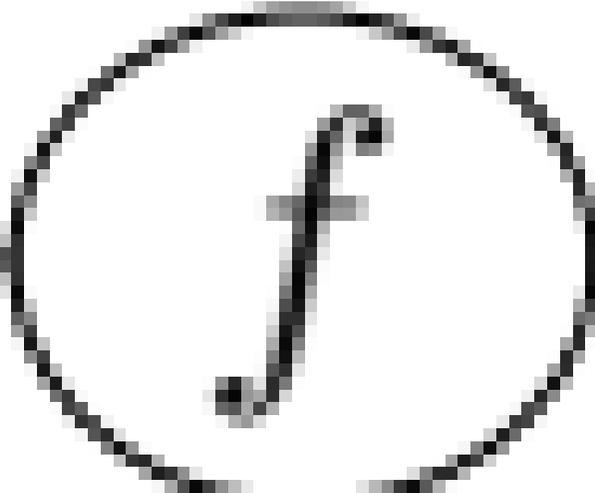
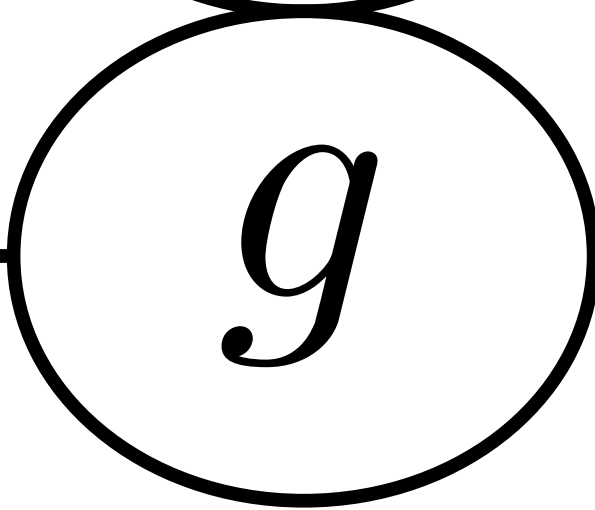
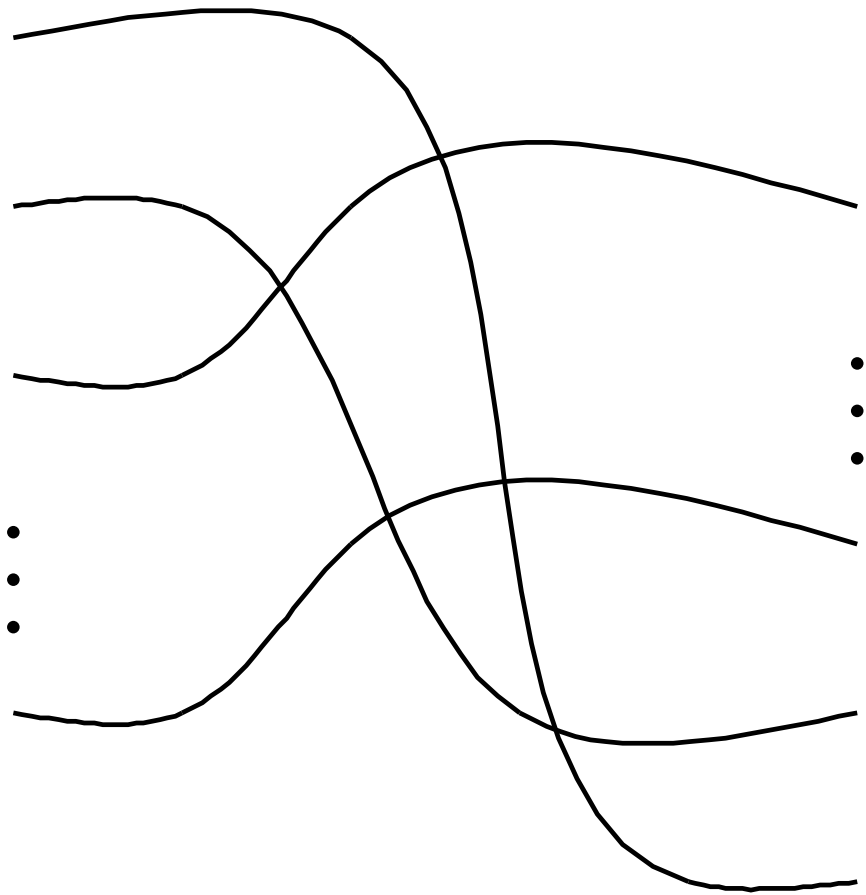
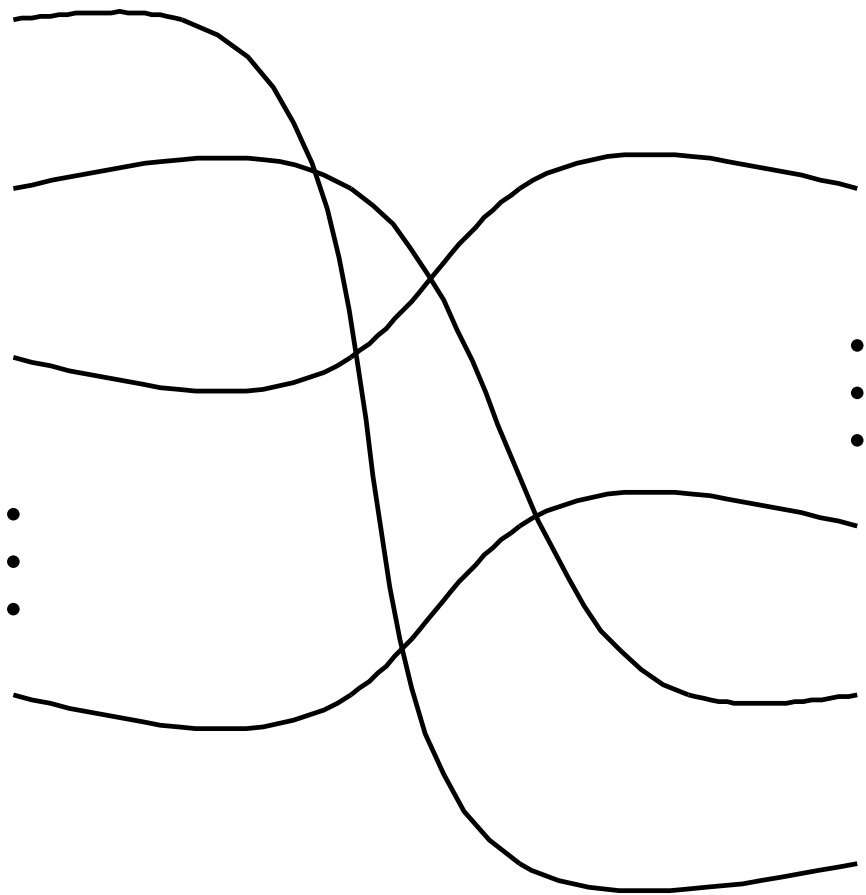
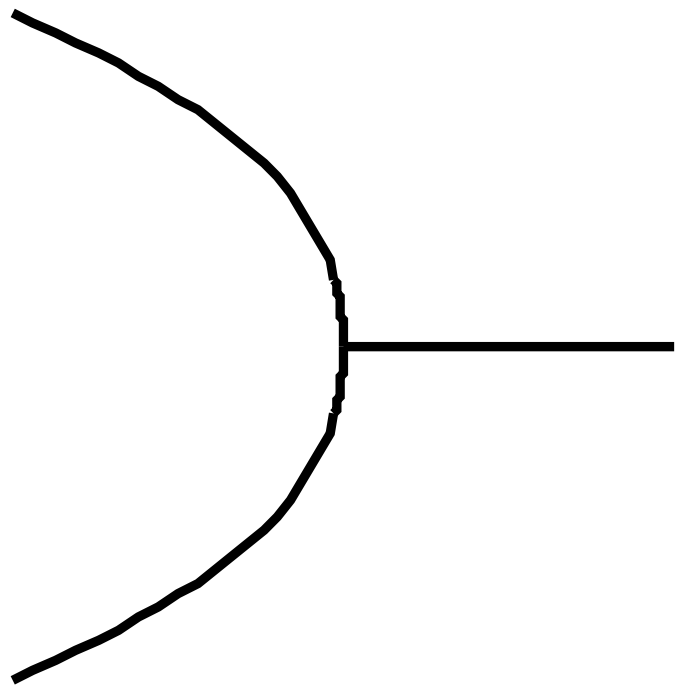
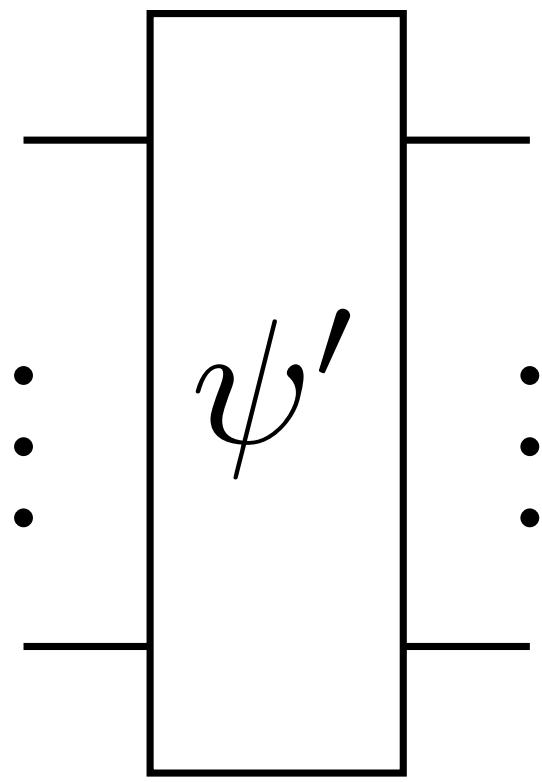
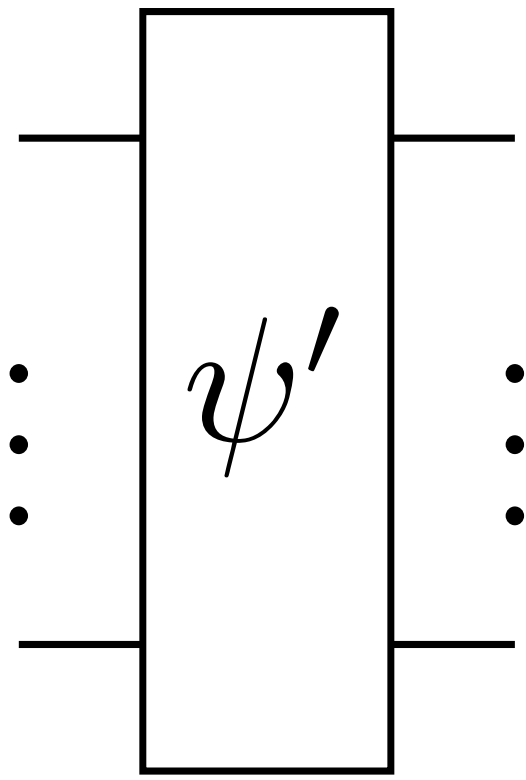
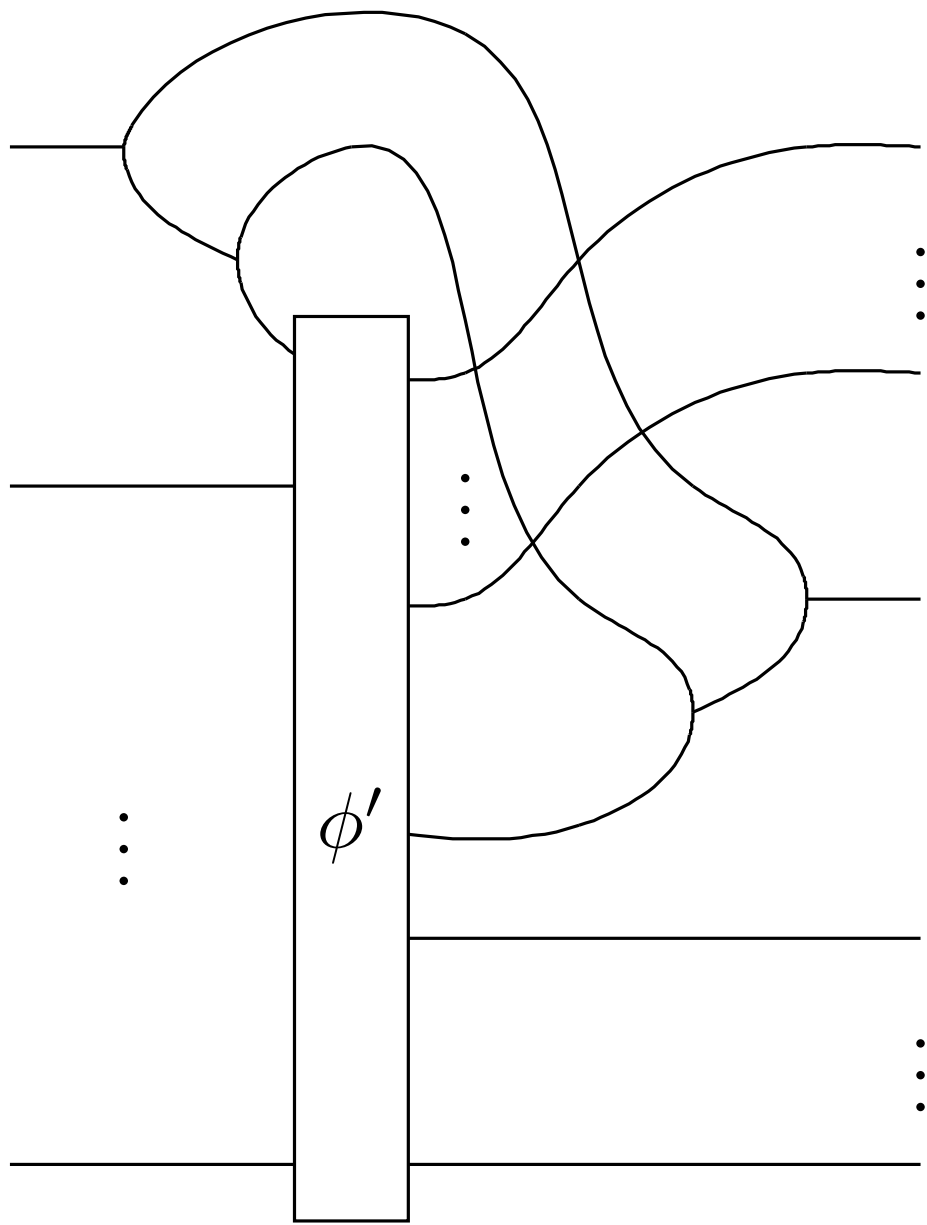
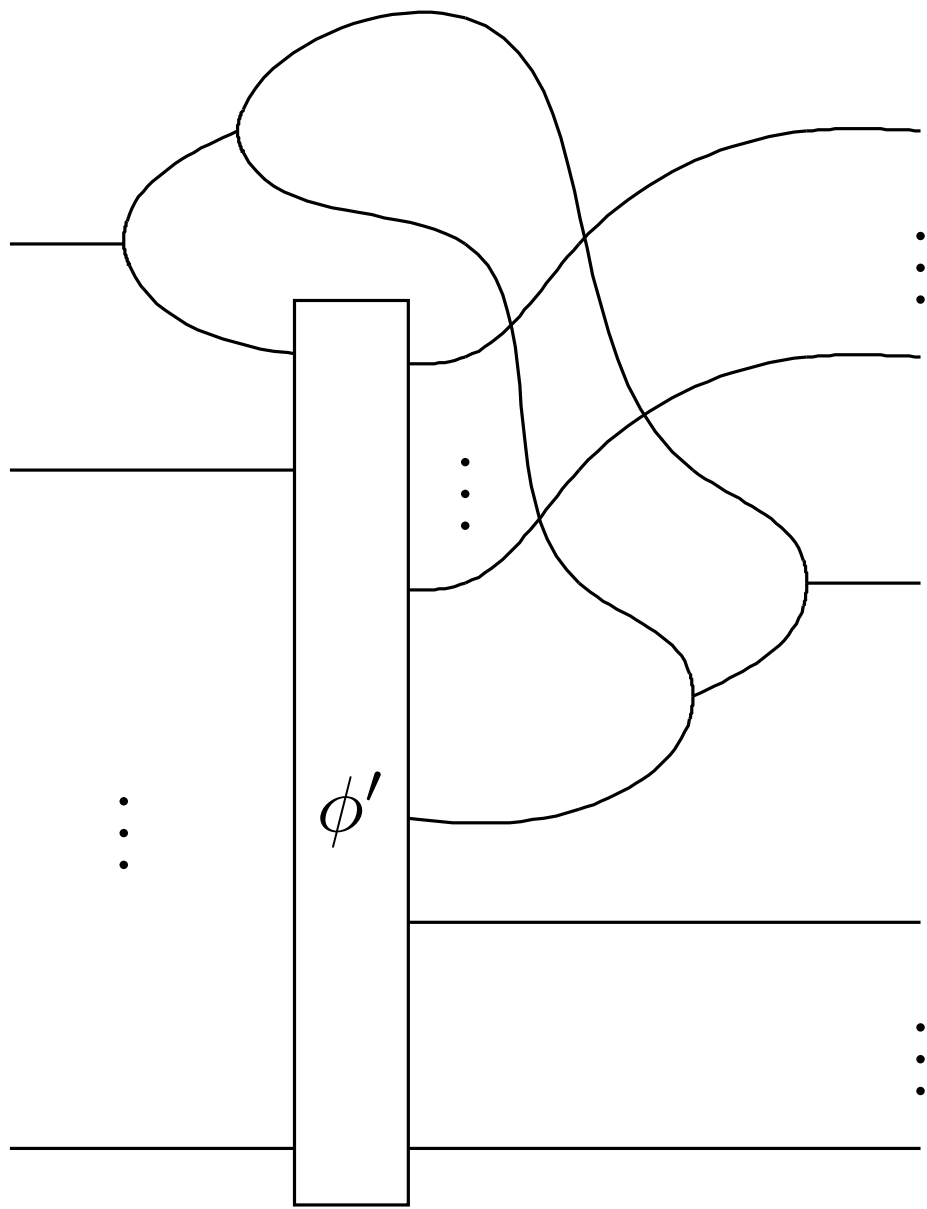
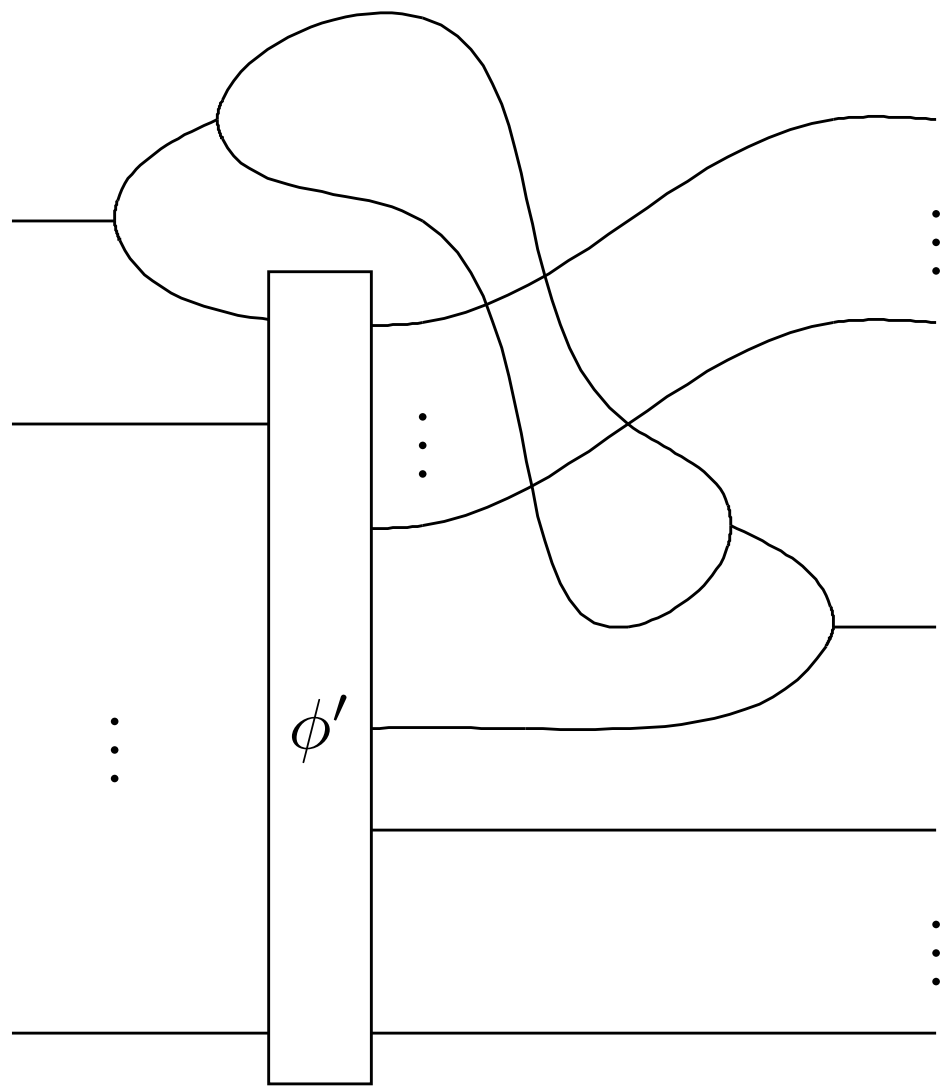
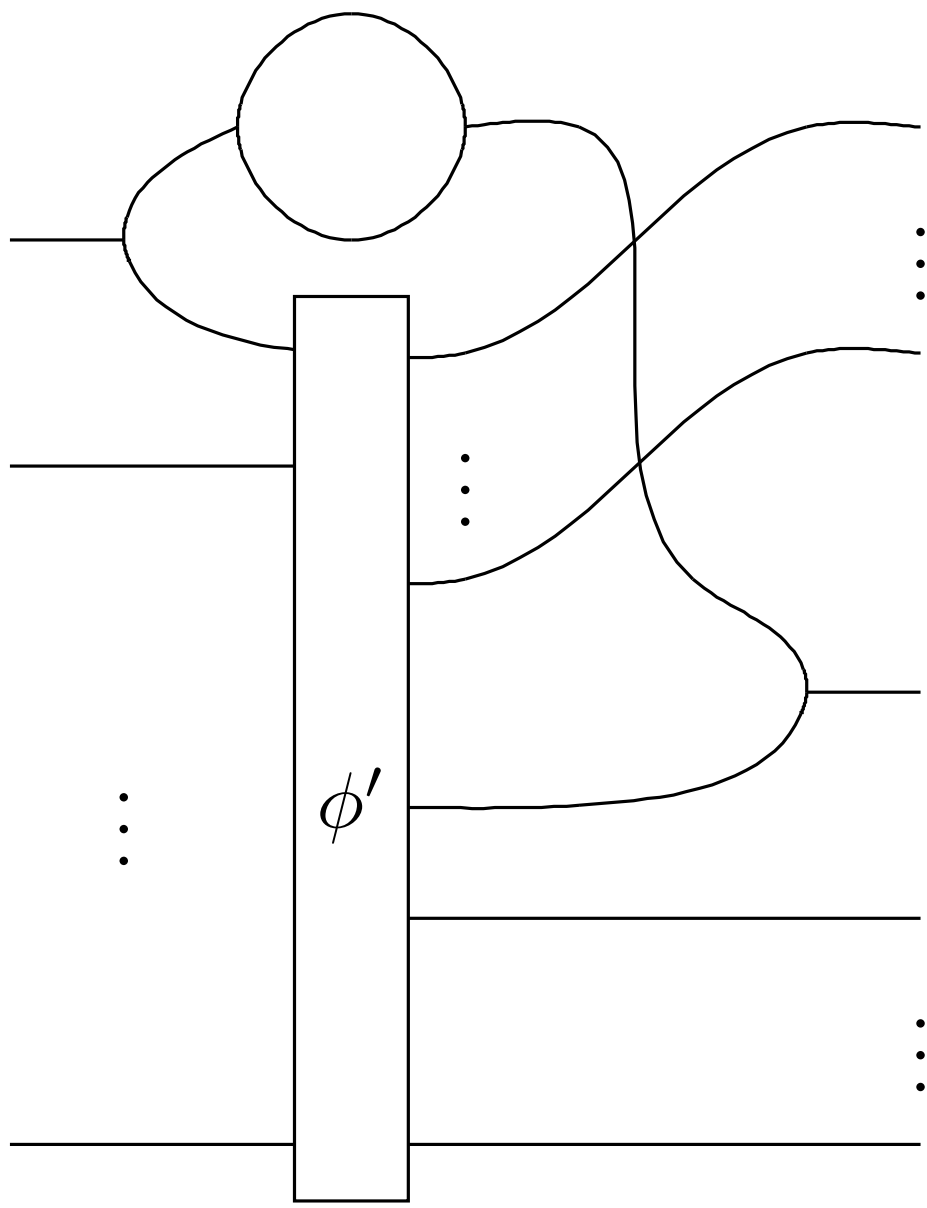
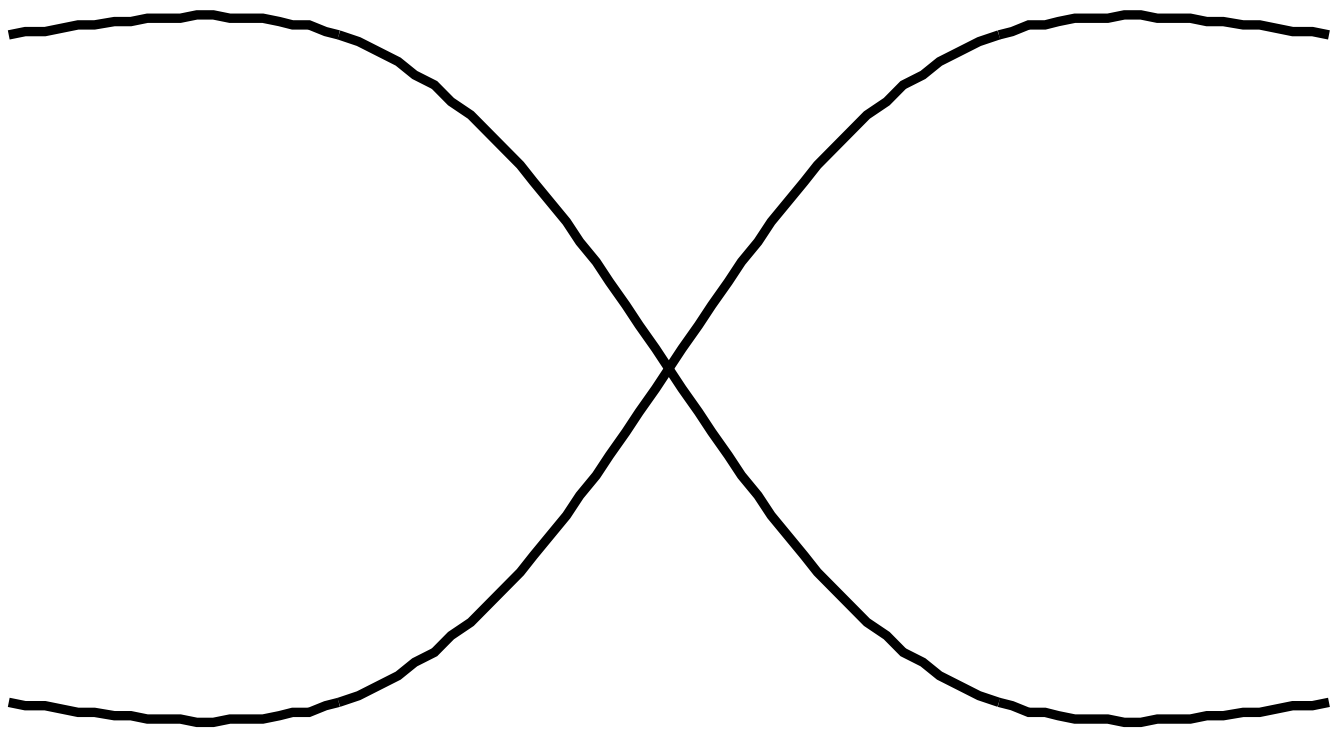
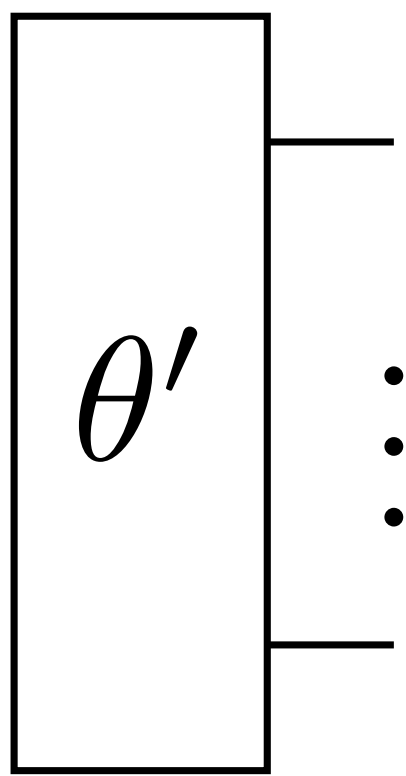
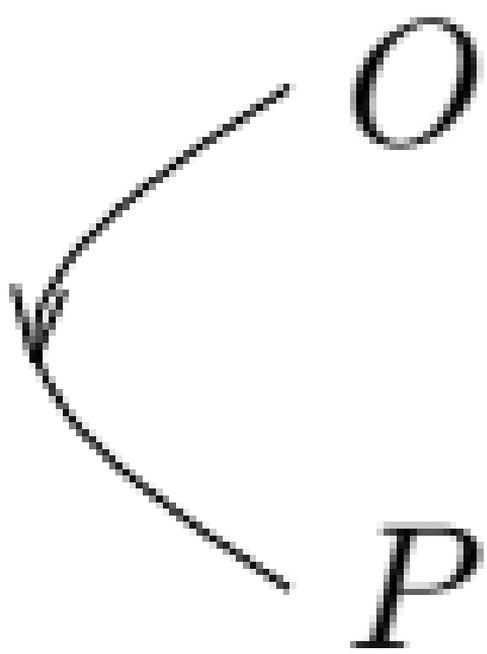
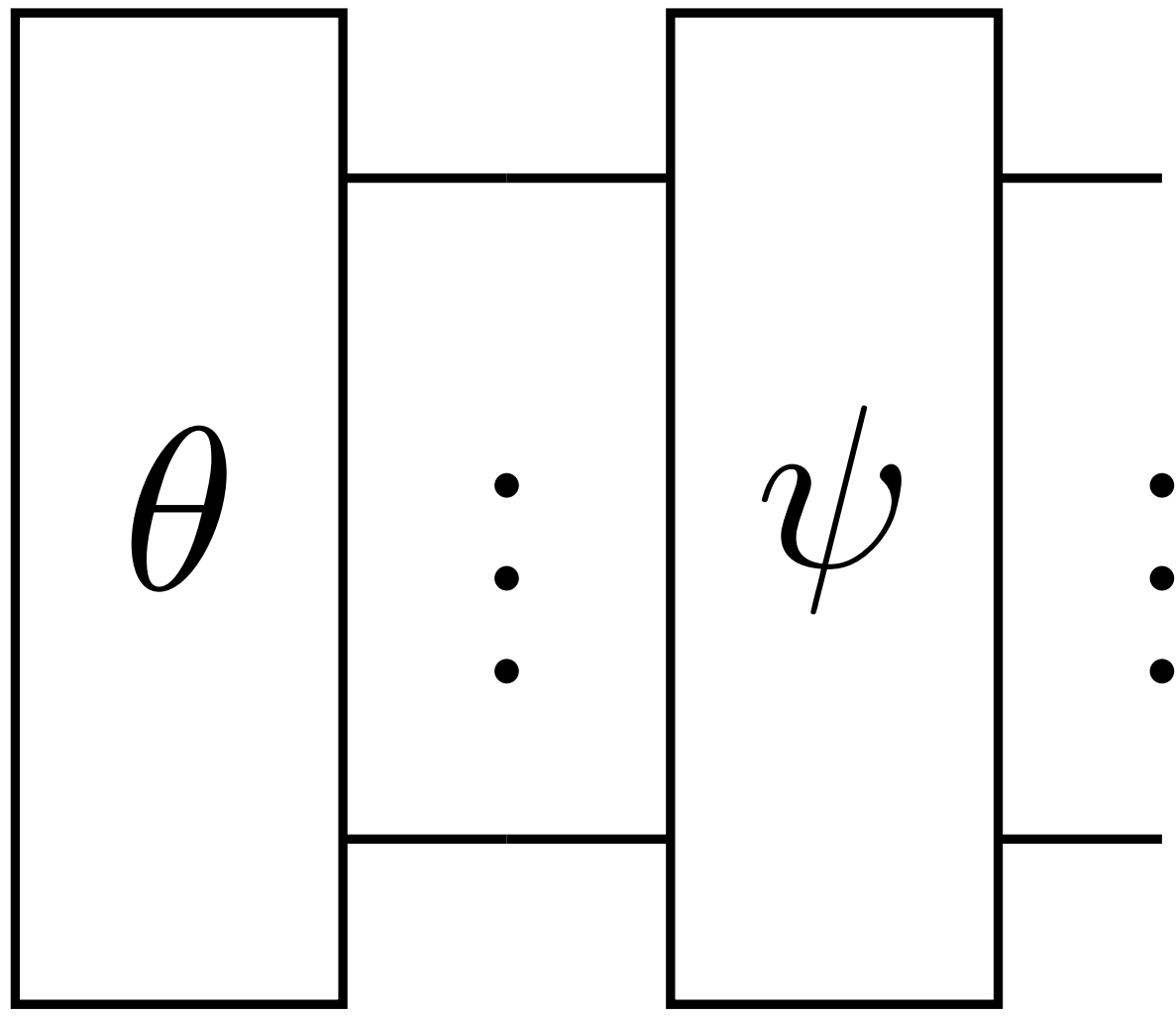
\end{document}